\documentclass[11pt,a4paper]{amsart}
\usepackage{amsmath,amsfonts,amsthm,amssymb}
\usepackage[T1]{fontenc}
\usepackage{color}
\usepackage{epsfig}
\usepackage{graphicx}
\usepackage{cite,url}
\usepackage{hyperref}
\usepackage{xspace}
\usepackage{enumerate}
\usepackage[left=1in,right=1in,top=1in,bottom=1in,foot=0.5in]{geometry}
\usepackage{thm-restate}
\usepackage{soul}
\usepackage{cleveref}

\newtheorem{theorem}{Theorem}

\newtheorem{proposition}[theorem]{Proposition}

\newtheorem{lemma}[theorem]{Lemma}

\newtheorem{case}{Case}

\theoremstyle{definition}

\newtheorem{definition}[theorem]{Definition}

\DeclareMathOperator{\next}{succ}
\DeclareMathOperator{\precc}{prec}
\DeclareMathOperator{\oG}{\overrightarrow{G}}

\newcommand{\tr}{_{\scalebox{0.53}{$\triangle$}}}
\newcommand{\trsq}{_{\scalebox{0.53}{$\triangle\square$}}}

\newcommand{\sq}{_{\scalebox{0.53}{$\square$}}}

\newcommand{\rDGz}{\textrm{(DG$_0$)}\xspace}
\newcommand{\rDG}{\textrm{(DG)}\xspace}
\newcommand{\rDM}{\textrm{(DM)}\xspace}

\newcommand{\rDTC}{\textrm{(DTC)}\xspace}
\newcommand{\rDB}{\textrm{(DB)}\xspace}
\newcommand{\rDH}{\textrm{(DH)}\xspace}

\newcommand{\rMG}{\textrm{(MG)}\xspace}
\newcommand{\rMGt}{\textrm{(MG$\tr$)}\xspace}
\newcommand{\rMGs}{\textrm{(MG$\sq$)}\xspace}
\newcommand{\rMTC}{\textrm{(MTC)}\xspace}
\newcommand{\rMB}{\textrm{(MB)}\xspace}
\newcommand{\rMH}{\textrm{(MH)}\xspace}
\newcommand{\rMM}{\textrm{(MM)}\xspace}

\newcommand{\cA}{\ensuremath{\mathcal{A}}\xspace}
\newcommand{\cB}{\ensuremath{\mathcal{B}}\xspace}

\newcommand{\cG}{\ensuremath{\mathcal{G}}\xspace}
\newcommand{\cP}{\ensuremath{\mathcal{P}}\xspace}

\newcommand{\N}{\ensuremath{\mathbb{N}}\xspace}

\author[J.\ Chalopin]{J\' er\'emie Chalopin}
\address{CNRS, Aix-Marseille Universit\'e, LIS, Marseille, France
}
\email{jeremie.chalopin@lis-lab.fr}

\author[V.\ Chepoi]{Victor Chepoi}
\address{Aix-Marseille Universit\'e, CNRS, LIS, Marseille, France}
\email{victor.chepoi@lis-lab.fr}

\author[M. Kokkou]{Maria Kokkou}
\address{Paderborn University}
\email{maria.kokkou@uni-paderborn.de}

\begin{document}

\pagestyle{plain}

\title{Distance-based certification for leader election in meshed graphs and local recognition of their subclasses}

\begin{abstract}
  In this paper, we present a 2-local proof labeling scheme with
  labels in $\{ 0,1,2\}$ for leader election in anonymous meshed graphs. Meshed
  graphs form a general class of graphs defined by a distance
  condition. They comprise several important classes of graphs, which
  have long been the subject of intensive studies in metric graph
  theory, geometric group theory, and discrete mathematics: median
  graphs, bridged graphs, chordal graphs, Helly graphs, dual polar
  graphs, modular, weakly modular graphs, and basis graphs of
  matroids.  We also provide 3-local proof labeling schemes to
  recognize these subclasses of meshed graphs using labels of size
  $O(\log D)$ (where $D$ is the diameter of the graph).

  To establish these results, we show that in meshed graphs, we can
  verify locally that every vertex $v$ is labeled by its distance
  $d(s,v)$ to an arbitrary root $s$.  To design proof labeling schemes
  to recognize the subclasses of meshed graphs mentioned above, we use
  this distance verification to ensure that the triangle-square
  complex of the graph is simply connected and we then rely on
  existing local-to-global characterizations for the different classes
  we consider.

  To get a proof-labeling scheme for leader election with labels of
  constant size, we then show that we can check locally if every $v$
  is labeled by $d(s,v) \pmod{3}$ for some root $s$ that we designate
  as the leader.
\end{abstract}

\maketitle

\section{Introduction} In local certification \cite{feuilloley2019introduction}---a broad term encompassing several models such as proof labeling schemes \cite{korman2010proof}, locally checkable proofs \cite{goos2016locally}, and nondeterministic local decision---the goal is to verify global properties through local conditions. Each node of a network/graph individually assesses these conditions and either accepts (i.e., if at the node all conditions are satisfied) or rejects (i.e., if at least one condition is violated). An error in the graph is detected if some node rejects. For example, to verify that a connected graph is a cycle and not a path, each node accepts if its degree is two and rejects if its degree is one. However, the inverse problem of accepting a path and rejecting a cycle is not possible without additional information, even if the nodes are not identical~\cite{feuilloley2019introduction}.

The main problem we consider in this paper is Leader Election \cite{LeLann77}. In anonymous networks, solving leader election is impossible without additional assumptions due to symmetry \cite{angluin1980local}. This impossibility result can be overcome by using randomization \cite{itai1990symmetry} or by restricting attention to specific graph classes \cite{yamashita1996computing}. Verifying that a unique leader exists in a graph where each node is equipped with a unique identifier can be done by encoding node identifiers or distances in the graph in the certificate, with certificates of size $O(\log n)$ (where $n$ is the number of vertices  of the graph) and $O(\log D)$ (where $D$ is the diameter of the graph), respectively. Verifying the existence of a unique leader (also known as the conjunction of \emph{At Most One Selected} (AMOS) \cite{fraigniaud2013amos} and \emph{At Least One Selected} (ALOS) \cite{feuilloley2018alos}) in anonymous graphs has recently started receiving attention \cite{chalopin2026stabilising,ChKo,FeSeSl}. In particular, \cite{FeSeSl} was the motivation for our work. Therein, Feuilloley, Sedl\'{a}\v{c}ek, and Sl\'{a}vik presented local certificates with labels of size $O(\log D)$ for leader election in anonymous chordal and grid graphs $G$ of diameter $D$. In their model, the label of each vertex $v$ is the presumed distance from $v$ to one and the same vertex $s$ (the root), not necessarily known to the vertices of $G$. The certification consists in the verification of a set of local  rules. If all vertices of $G$ pass this verification, then the labels  of vertices $v$ are the distances $d(s,v)$ from $v$ to  $s$ and $s$ can be elected  as the leader.
In \cite{chalopin2026stabilising,ChKo}, certificates were designed for
simply connected subgraphs of the triangular
grid~\cite{chalopin2026stabilising} as well as for chordal and
$K_4$-dismantlable graphs~\cite{ChKo}. These certificates are obtained
by orienting all edges and designing local verification rules such
that if these rules are satisfied at every vertex, then the
orientation of the graph contains a unique sink that can be designated
as a leader. This gives certificates of size $O(\Delta)$ where
$\Delta$ is the maximum degree of the graph. 
As a form of fault-tolerance, verification is a recurrent but not the main task to be solved within a system and as a result it should not consume too many resources. Hence, in this work, we present constant size certificates for large classes of graphs, namely for meshed graphs and their subclasses. Similarly to \cite{FeSeSl}, we begin by establishing certificates based on distances to the root that are logarithmic in the diameter of the graph (Theorems \ref{meshed1-intro} and \ref{bridged-helly1-intro}). Then we extend this method and show how to transform these certificates to certificates of constant size by replacing the distances to the root  by these  distances modulo 3. We show that these certificates ensure the existence and the election of a unique leader in all meshed graphs (Theorem \ref{meshed2-intro}). It should be noted, that the class of meshed graphs is very general and comprises most classes of graphs studied in metric graph theory.

In addition to Leader Election, we show how to use our first result
and the already known local-to-global characterizations of all main
classes of meshed graphs (but not for the class of all meshed graphs)
in order to construct certificates verifying that a given graph
belongs to one of the considered classes (Theorem
\ref{local-recogn-intro}). This is particularly useful, as it is
possible that an algorithm only works in a specific class of graphs
and becomes stuck in others. Thus, being able to locally verify that
the graph in which an algorithm is run belongs to the intended class
can lead to avoiding deadlocks. This direction has also received
significant attention in the literature.  Previous works have
established that several classes of graphs can be recognized with
certificates of size $O(\log n)$ (where $n$ is the number of vertices
of the graph): e.g., planar graphs~\cite{feuilloley2021compact},
graphs of bounded
genus~\cite{esperet2022certification,feuilloley2023genus}, and several
geometric intersection classes of graphs (chordal, interval, circular
arc, trapezoid, and permutation)~\cite{jauregui2025compact}. On the
other hand, there exist other geometric intersection graph classes
that cannot be recognized with certificates of size
$o(n)$~\cite{defrain2024local}.

\subsection{Our results}
In this paper, we significantly extend, unify, and improve the results
of \cite{FeSeSl}. Namely, we consider the following \emph{local
  distance certification problem}: given a class of graphs
$\cG$, we ask for a set of $r$-local rules
$\cA$ ($\cA$ is called an \emph{$r$-local verifier}) for $\cG$ such that, given any anonymous graph $G=(V,E)$ from
$\cG$ and any set $\{ D(v): v\in V\}$ of labels of vertices of
$G$, all vertices of $G$ satisfy the local rules of $\cA$ if and only if 
there exists a root vertex $s$ in $G$ such that for any $v\in V$ the
label $D(v)$ is the distance $d(s,v)$ from $s$ to $v$.  By
$r$-\emph{locality} is meant that the rules at $v$ depend only on the
labels of the vertices in the $r$-ball $B_r(v)$ centered at $v$ and
the subgraph induced by $B_r(v)$. Our first main result is the
following theorem:

\begin{theorem}\label{meshed1-intro}
  The class of meshed graphs admits a 2-local distance certification
  protocol.\end{theorem}

Meshed graphs are defined by a condition on the distance function (the
Weak Quadrangle Condition) and they comprise numerous important
subclasses of graphs: weakly modular graphs, modular graphs, bridged
graphs, Helly graphs, median graphs, dual polar graphs, basis graphs
of matroids, etc. All these classes are among the main classes of
graphs studied in metric graph theory \cite{BaCh_survey,CCHO} and also
arise in geometric group theory, metric geometry, combinatorial
optimization, and concurrency. Most of these classes are weakly
modular graphs, which are defined by two conditions, the Triangle and
the Quadrangle Conditions.  Bridged and median graphs are 
far reaching generalizations of chordal graphs and grids
considered in~\cite{FeSeSl}. 

In the case of (weakly) bridged and Helly graphs, Theorem
\ref{meshed1-intro} can be slightly improved:

\begin{theorem}\label{bridged-helly1-intro}
  The classes of (weakly) bridged graphs and Helly graphs admit
  1-local distance certification protocols.
\end{theorem} 

Our second result concerns the \emph{local distance mod 3
  certification problem} and its use to the \emph{leader election
  problem}. The local distance mod 3 certification problem is similar
to the local distance certification problem, only the labels $L(v)$ of
vertices $v$ belong to the set $\{ 0,1,2\}$ and all vertices $v$ of
$G$ satisfy the local rules of $\cA$ if and only if
$L(v) = d(s,v) \pmod{3}$ from some root $s$ to $v$. From a local
verifier for the local distance mod 3 certification problem, one can
derive a leader election labeling scheme, i.e., a proof labeling
scheme that certifies the existence of a unique leader in the graph.

\begin{theorem}\label{meshed2-intro}
 The class of meshed graphs admits   a $2$-local verifier for
  the distance mod 3 certification problem.
Consequently, the class of meshed graphs has a $2$-local leader election labeling
  scheme with labels of constant size. \end{theorem}

Theorem~\ref{meshed2-intro} improves and generalizes the result
of~\cite{FeSeSl} as it considers a much larger
class of graphs and since it uses a constant number of labels.  The
class of simply connected subgraphs of the triangular grid considered
in~\cite{chalopin2026stabilising} is a subclass of $K_4$-free bridged
graphs which is in turn a subclass of $K_4$-free dismantlable graphs
considered in~\cite{ChKo}, as well as a subclass of meshed graphs
considered in the current paper.  Dismantlable graphs, also known as
cop-win graphs~\cite{NoWi,Qui83}, also contain several important
classes from metric graph theory such as bridged or Helly
graphs. However the class of $K_4$-free dismantlable graphs considered
in~\cite{ChKo} and the class of meshed graphs are incomparable.

For bridged and Helly graphs, Theorem \ref{meshed2-intro} can be
refined in the following way:

\begin{theorem}\label{bridged-helly2-intro}
  The classes  of (weakly) bridged graphs and Helly graphs admit 
  a $1$-local verifier for the distance mod 3 certification
  problem.
Consequently, there exists a $1$-local leader election labeling
  scheme using labels of constant size for these classes of graphs.
\end{theorem} 

In order to prove Theorems~\ref{meshed2-intro}
and~\ref{bridged-helly2-intro}, we associate a directed graph $\oG_L$
to any graph $G$ and any labeling $L: V(G) \to \{0,1,2\}$. The vertex
set of $\oG_L$ is $V(G)$ and there is an arc $\overrightarrow{vu}$ in
$\oG_L$ if and only if $vu$ is an edge of $G$ and
$L(v)=L(u) +1 \pmod{3}$. In order to prove
Theorems~\ref{meshed2-intro} and~\ref{bridged-helly2-intro}, we prove
that if the local verifier accepts at every vertex of $G$, then $\oG_L$
contains a unique sink that can be designated as the leader.

\begin{proposition}\label{prop-meshed2-intro}
  For any meshed (respectively, weakly bridged or Helly) graph $G$ and
  any labeling $L= V(G) \to \{0,1,2\}$ satisfying the rules of the
  local verifier of Theorem~\ref{meshed2-intro} (respectively, of
  Theorem~\ref{bridged-helly2-intro}), the following hold:
  \begin{enumerate}[(1)]
  \item the directed graph $\oG_L$ is acyclic;
  \item if $s$ is any sink of $\oG_L$, then $L(v)=d(s,v)\pmod{3}$ for any vertex $v\in V$;
  \item $\oG_L$ has a unique sink $s$.\end{enumerate}
\end{proposition}

The main idea behind the local certification protocols for distance
and distance mod 3 problems in meshed, bridged and Helly graphs, is to
replace the distance $d(s,v)$ in Triangle and Weak Quadrangle
Conditions by the labels $D(v)$ and $L(v)$. This allows to prove
Theorems~\ref{meshed1-intro} and \ref{bridged-helly1-intro}, as well
as Assertion (2) in \Cref{prop-meshed2-intro}. The proof of Assertion
(1) of \Cref{prop-meshed2-intro} uses the fact that the
triangle-square complexes of meshed graphs are simply
connected~\cite{CCHO} and relies on a generalization of the Sperner's
lemma to disk triangulations due to Musin~\cite{Musin}. Assertion (1)
holds for any graph with a simply connected triangle-square complex,
and thus, it generalizes a result of~\cite{ChKo} established for
dismantlable graphs.

Our third result concerns local recognition of classes of graphs based
on local distance certificates.  A class of graphs $\cG$ admits an
\emph{$r$-local recognition protocol of size $k$} if there exists a
set $S \subseteq \{0,1\}^k$ of labels and an $r$-local verifier $\cA$
(for labeled graphs with labels in $S$) such that a graph $G \in \cG$
if and only if there is a labeling $R: V(G) \to S$ such that the rules
of $\cA$ are satisfied at every vertex $v \in V(G)$. We show that all
the subclasses of meshed graphs we mentioned above admit a local
recognition protocol of size $O(\log D)$ (where $D$ is the diameter of
the graph).

\begin{theorem}\label{local-recogn-intro} Each of the following
  classes of meshed graphs admits an $r$-local recognition protocol
  with labels of size $O(\log D)$:
  \begin{enumerate}[(1)]
  \item chordal and  bridged graphs  with $r=1$,
  \item weakly bridged graphs and Helly graphs with $r=2$,
  \item weakly modular, modular, median, pseudo-modular, bucolic, dual
    polar, and sweakly modular graphs with $r=3$,
  \item basis graphs of matroids and of even $\Delta$-matroids with
    $r=3$.
  \end{enumerate}
\end{theorem} 

As explained above, one can recognize cycles without certificates
(i.e., $k = 0$). But it is well-known
(see~\cite{feuilloley2019introduction} for example) that 
for any fixed radius $r$ and any fixed $p \geq 2r+2$, in order to
recognize all paths and reject the cycle $C_p$ of length $p$ with an $r$-local recognition protocol, one needs labels of size $\Omega(\log n) = \Omega(\log
D)$. This is due to the fact that in $C_p$, all $r$-balls are paths. This implies that in order to recognize a class of graphs that
contains all paths but no long cycles, one needs to use labels of size
$\Omega(\log D)$. This lower bound then applies to the class of
meshed graphs, as well as to all the subclasses mentioned in Items (1),
(2), and (3) in Theorem~\ref{local-recogn-intro}, establishing that
the sizes of the labels we use in these cases are optimal.

Our local recognition method is based on two ingredients: simple
connectivity of the triangle-square complexes of all meshed graphs and
local-to-global characterizations of many subclasses of meshed
graphs. These local-to-global characterizations are all of the form: a
graph $G$ belongs to a class $\cG$ if and only if its
triangle-square complex is simply connected and $G$ satisfies a
specific local condition (say, all small balls of $G$ look like balls
in graphs from $\cG$). Such type of results are called Cartan-Hadamard
theorems and originate from Riemannian geometry.  While general meshed
graphs do not admit local-to-global characterizations~\cite{CCHO}, all
their subclasses mentioned in this paper can be characterized in this
way. This is the fruit of a long line of research, which started with
the paper \cite{Ch_CAT} and culminated with the papers \cite{ChChOs}
and \cite{CCHO}. Checking if the triangle-square complex $X\trsq(G)$
of a given graph $G$ is simply connected is an undecidable
problem~\cite{Hak}. However, we show that if the local verifier of
Theorem~\ref{meshed1-intro} accepts a graph $G$ with some labeling
$L: V(G) \to \N$, then $X\trsq(G)$ is simply connected. We also know
that for any meshed graph $G$, $G$ is recognized by this verifier when
$L = d(s,\cdot)$ for some vertex $s$ of $G$. Consequently, we have a
proof labeling scheme that accepts all meshed graphs, rejects all
graphs that do not have a simply connected triangle-square complex,
and behaves on an undefined way on graphs  that are not meshed but
have simply connected triangle-square complexes. This property is then
sufficient to obtain local recognition protocols for the metric
classes of graphs we consider from their existing local-to-global
characterizations.

\subsection{Related work} In this subsection, we briefly present  the related work on locally checkable labelings and proofs and the general framework and the main motivations behind meshed graphs and their subclasses.  

\subsection*{LCL and LCP} 
In the LOCAL model introduced by Linial~\cite{linial1987local,linial1992journal}, it is assumed that nodes exchange messages with neighboring nodes until a task is solved. The complexity of an algorithm within this model only depends on the number of synchronous communication rounds. An alternative definition of the model that is more appropriate for our setting, is to assume that every node can see the state of each node in its neighborhood up to some radius $r$ and based on this partial knowledge of the graph, choose an output. The goal in this case is to minimize $r$ and the size of states for some given task. What can be computed within this setting has received significant attention (e.g., \cite{vishkin1986determinisitic,goldberg1987parallel,luby1986simple,alon1986fast}). The concept of locally verifiable problems was formalized by Naor and Stockmeyer in their seminal paper~\cite{naor1993lcl,naor1995journal} in which they defined the family of \emph{Locally Checkable Labeling} (LCL) problems. A problem is in LCL if given the view, consisting of constant-sized input and output labels, at some constant radius $r$ for each node in a class of bounded degree graphs, there exists an algorithm that can decide whether a computed solution is correct. Although LCL problems were continually explored since their introduction, no groundbreaking results emerged until 2016 when Brandt et al.~\cite{brandt2016lower} led to a series of papers (e.g., \cite{balliu2018classes, balliu2018locally, balliu2020randomness}) that gradually provided a complete picture on the classification of the round complexity of all LCL problems in four broad classes \cite{suomela2020landscape}. Recently, Bousquet et al.~\cite{feuilloley2025certification} proposed a different direction of classification of LCL problems based on space instead of round complexity. In this paper we consider large classes of graphs from metric graph theory, namely \emph{meshed graphs} and their subclasses. In all our results, the radius $r$ is a small constant (usually, $r\in \{ 2,3\}$). Our approach is complementary to that of \cite{feuilloley2025certification} in the sense that we only consider specific problems but for much larger classes of graphs.

A generalization of LCL, called \emph{Locally Checkable Proofs (LCP)}, was introduced in \cite{goos2016locally}. The main difference between the two approaches is that the goal in LCL is to locally solve a global problem whereas in LCP (and in the related model of \emph{Proof Labeling Schemes (PLS)} \cite{korman2005proof,korman2010proof}) the solution to some problem can be computed globally and it is only \emph{verified} locally. That is, in LCP each node $v$ is provided with a certificate: if the certificate of $v$ and the certificates of the nodes in the graph induced by the $r$-neighborhood of $v$, for some constant $r$, satisfy a number of given conditions {(this is also called an \emph{$r$-local verifier})} then $v$ accepts the local solution. If every node in the graph also accepts, then the solution is globally correct. {Proof Labeling Schemes are closely connected to LCP with two key differences: (i) in the classical PLS model, each node only has access to certificates of its neighbors at distance one (as opposed to any constant $r$ in LCP) and (ii) the nodes do not have access to the identifiers and input of their neighbors. Unlike PLS and other models such as non-deterministic local decision}, in LCP the verification of a solution is {therefore} only limited by locality (i.e., the chosen radius $r$). In more recent work, such as \cite{FeSeSl} the concept of \emph{$r$-local} PLS is used, nullifying the first difference between the two models. When considering only anonymous nodes and when the local conditions do not depend on the input of the neighbors (as it is the case in the current paper or in~\cite{FeSeSl}), the concepts of LCP and $r$-local PLS can be used interchangeably. Finally, as LCP was introduced as a generalization of LCL, in the particular case where all certificates given to nodes are empty, LCP and LCL are identical.

\subsection*{Meshed graphs and their subclasses} The graphs studied in this paper are defined or characterized by metric properties. Some of these properties are discrete analogs of basic properties of classical metric spaces: Euclidean and hyperbolic spaces and their generalization---CAT(0) spaces (geodesic spaces of nonpositive global curvature), Manhattan $\ell_1$-spaces and hypercubes,  and $\ell_\infty$-spaces and their generalization---injective spaces. The resulting classes  of graphs share common metric properties, dubbed weak modularity and meshedness, which allow to study them from a unified point of view. Many of them have been characterized topologically, in a local-to-global way. This allowed to establish several correspondences between graph classes from metric graph theory and classes of simplicial and cube complexes studied in geometric group theory. 

For example, bridged graphs are the graphs in which the balls around
convex sets are convex \cite{FaJa,SoCh}. Convexity of balls and of
balls around convex sets are fundamental properties of convexity in
Euclidean spaces, which is also important in CAT(0) geometry
\cite{BrHa}. For graphs, this convexity property turned out to be
equivalent to yet another truly metric property: bridged graphs are
the graphs in which all isometric cycles have length 3
\cite{FaJa,SoCh}. Classical examples of bridged graphs are chordal
graphs and the triangular grid. A local-to-global characterization of
bridged graphs was given in \cite{Ch_CAT}: they are the graphs whose
clique complex is simply connected and that do not contain $4$-wheels
and $5$-wheels. These graphs were rediscovered~\cite{JaSw} in
geometric group theory as the $1$-skeletons of systolic complexes.
Similar correspondences hold for Helly and median graphs. Helly graphs
are the discrete analogs of $\ell_{\infty}$- and injective spaces
since they are the graphs in which balls satisfy the Helly property:
any collection of pairwise intersecting balls has a nonempty
intersection.Several metric characterizations of Helly graphs have
been obtained in \cite{BP-absolute,BaPr} and a local-to-global
characterization of clique complexes of Helly graphs was given in
\cite{CCHO}. Helly graphs are universal in the following sense: any
graph $G$ is an isometric subgraph of a Helly graph.
Finally, median graphs are the discrete analogs of
$\ell_1$-spaces. Examples of median graphs are grids, trees,
and hypercubes. They are precisely
the 1-skeleta of CAT(0) cube complexes~\cite{Ch_CAT,Ro} that are an
important object of study in geometric group theory~\cite{Gromov}. 

Albeit featuring three different types of classical metrics ($\ell_2$,
$\ell_1$, and $\ell_\infty$), bridged, median, and Helly graphs share
two common distance properties, called the Triangle Condition (TC) and
the Quadrangle Condition (QC) (see Section~\ref{sec-wm}). The graphs
satisfying (TC) and (QC) are called \emph{weakly
  modular}~\cite{BaCh_helly,Ch_metric}. They comprise several other
interesting classes of graphs, such as modular, pseudo-modular, dual
polar, sweakly modular, and bucolic graphs. \emph{Meshed graphs} were
introduced in~\cite{BaMuSo}, as a generalization of weakly modular
graphs and are defined by a weaker version (QC$^-$) of the Quadrangle
Condition.  Though meshed graphs have been introduced to capture
distance properties of the icosahedron~\cite{BaMuSo}, two important classes of
graphs have been shown in \cite{Ch_delta} to be meshed (but not weakly
modular): basis graphs of matroids and of even $\Delta$-matroids. Even
if meshed graphs have simply connected triangle-square complexes,
it was shown in \cite{CCHO} that general meshed graphs cannot be
characterized in a local-to-global way. However, all the subclasses of
meshed graphs mentioned above admit such
characterizations~\cite{BrChChGoOs,CCHO,ChOs} that are at the heart of
their local recognition (Theorem \ref{local-recogn-intro}).

\section{Preliminaries}

\subsection{Graphs}
All graphs $G=(V(G),E(G))$ considered in this paper are finite,
undirected, connected, and simple. For two distinct vertices
$u,v\in V$ we write $u\sim v$ if $u$ and $v$ are connected by an edge
and $u\nsim v$ otherwise.  As usual, $N(v)$ denotes the set of
neighbors of a vertex $v$ in $G$ and $N[v]=N(v)\cup \{ v\}$.  The
subgraph of $G$ \emph{induced by} a subset $A\subseteq V(G)$ is the
graph $G[A]=(A,E')$ such that $uv\in E'$ if and only if $uv\in
E(G)$. A \emph{triangle} of $G$ is a cycle of length $3$ and a
\emph{square} of $G$ is an induced cycle of size $4$.  The
\emph{distance} $d(u,v)=d_G(u,v)$ between two vertices $u$ and $v$ of
$G$ is the length of a shortest $(u,v)$--path.  
For a vertex $v$ of $G$ and an integer $r\ge 1$, we will denote by
$B_r(v)$ the \emph{ball} in $G$ of radius $r$ centered at $v$, i.e.,
$B_r(v)=\{ x\in V: d(v,x)\le r\}$.  More generally, the
$r$--\emph{ball around a set} $A\subseteq V(G)$ is the set
$B_r(A)=\{ v\in V(G): d(v,A)\le r\},$ where
$d(v,A)=\text{min} \{ d(v,x): x\in A\}$. The \emph{interval}
$I(u,v)$ between $u$ and $v$ consists of all vertices on shortest
$(u,v)$--paths.  A subgraph
$H=(V(H),E(H))$ is an \emph{isometric} subgraph of $G=(V(G),E(G))$ if
$V(H) \subseteq V(G)$ and for any $u,v \in V(H)$, $d_H(u,v) = d_G(u,v)$ 
and is a \emph{convex} subgraph of $G$ if $I(u,v)\subseteq V(H)$ for 
any two vertices $u,v\in V(H)$.

\subsection{Cell complexes of graphs}

All cell complexes considered in this paper are
CW-complexes~\cite{Ha}. Given a cell complex $X$, its $1$-skeleton
$G(X)$ is the graph whose vertices are the $0$-cells of $X$ and the
edges are the $1$-cells of $X$. More generally, the
\emph{$k$-skeleton} $X^{(k)}$ of $X$ is the cell complex containing
only the cells of $X$ of dimension at most $k$. As morphisms between
cell complexes we always consider \emph{cellular maps}, that is, maps
sending the $k$--skeleton into the $k$--skeleton.

For a graph $G$, we define its \emph{triangle} (respectively,
\emph{square}) complex $X\tr(G)$ (respectively, $X\sq(G)$) as a
two-dimensional cell complex where the $0$-cells are the vertices of
$G$, the $1$-cells are the edges of $G$ and the $2$-cells are (solid)
triangles (respectively, squares) whose boundaries are identified by
isomorphisms with (graph) triangles (respectively, squares) in $G$. A
\emph{triangle-square complex} $X\trsq(G)$ is defined analogously, as
the union of $X\tr(G)$ and $X\sq(G)$.  A triangle-square complex is
\emph{flag} if it coincides with the triangle-square complex of its
1-skeleton.

A cell complex $X$ is called \emph{simply connected} if it is
connected and if every continuous mapping of the 1-dimensional sphere
$S^1$ into $X$ can be extended to a continuous mapping of the disk
$D^2$ with boundary $S^1$ into $X$. Note that $X$ is connected iff
$G(X)$ is connected, and $X$ is simply connected iff its 2-skeleton
$X^{(2)}$ is simply connected. Equivalently, a cell complex $X$ is
simply connected if $X$ is connected and every cycle $C$ of its
1-skeleton is null-homotopic, i.e., it can be contracted to a single
point by elementary homotopies. This is also equivalent to the
existence of disk diagrams, which we define next. Let $X$ be a cell
complex and $C$ be a cycle in the $1$--skeleton of $X$. Then a pair
$(D, \varphi)$ is called a \emph{disk diagram} (or Van Kampen diagram)
if $D$ is a plane graph and $\varphi: D\rightarrow X$ is a cellular
map such that $\varphi$ maps the boundary of $\partial D$ to $C$ (for
more details see \cite[Chapter V]{LySch}).  According to Van Kampen's
lemma \cite[pp.\ 150--151]{LySch}, for every cycle $C$ of a simply
connected cell complex, one can construct a disk diagram.
If $X$ is a simply connected
triangle-square complex, then for each cycle $C$ all inner faces in a
disk diagram $(D,\varphi)$ of $C$ are triangles or squares and the
image of each triangle (respectively, square) of $D$ is a triangle
(respectively, a square) of $X$.

\subsection{Proof Labeling Schemes}

We follow the notations of~\cite{FeSeSl} (see
also~\cite{goos2016locally}). A labeled graph $(G,L)$ is a graph where
each vertex $v$ is initially given a label $L(v)$.  For a family of
(labeled) graphs $\cG$, a \emph{property} $\cP$ is a subset of $\cG$ that is
closed under isomorphism.

A \emph{certificate assignment} (or a proof) $P$ for a graph
$G \in \cG$ is a function $P: V(G) \to \{0,1\}^*$ that associates with
each vertex a \emph{certificate}. The certificate is of size $k \in \N$
if $P(v) \leq k$ for every $v \in V(G)$. A \emph{verifier} $\cA$ is a
function that take a graph $G \in \cG$, a certificate assignment $P$
for $G$, and a vertex $v$ of $G$ and outputs $\cA(G,P,v) \in
\{0,1\}$. A verifier $\cA$ is \emph{$r$-local} if the output of $\cA$
depends only on the vertices located in the ball of radius $r$ around $v$, i.e.,
$\cA(G,P,v) = \cA(B_r(v), P_{|B_r(v)},v)$ for every $G \in \cG$ and
every $v \in V(G)$.

A graph property $\cP \subseteq \cG$ admits a \emph{proof labeling
  scheme} of radius $r$ and of size $\kappa: \cP \to \N$ in the family $\cG$
if there is an $r$-local verifier $\cA$ such that the following
properties hold:
\begin{enumerate}[(i)]
\item \emph{Completeness}: for any $G \in \cP$, there exists a certificate
  assignment $P$ for $G$ of size $\kappa(G)$ such that
  $\cA(G,P,v) = 1$ for each $v \in V(G)$.
\item \emph{Soundness}: for any $G \in \cG \setminus \cP$, for any
  certificate assignment $P$ for $G$, there exists $v \in V(G)$ such
  that $\cA(G,P,v) = 0$. 
\end{enumerate}

In this paper, we provide proof labeling schemes for some recognition
problems in general graphs and for leader election in specific classes
of graphs.  A family of graphs $\cP$ admits an \emph{$r$-local
  recognition protocol of size $k$} if there exists an $r$-local proof
labeling scheme of size $k$ for the class $\cP$ in the family of all
unlabeled graphs.

In order to express the leader election problem as a property, we
assume that each vertex $v$ is initially assigned a label
$I(v) \in \{0,1\}$ and we want to recognize if there is exactly one
vertex in the graph that is labeled by $1$. We say that a family of
graphs $\cG$ admits an \emph{$r$-local leader election labeling scheme
  of size $s$} if for the family of graphs $\cG$ where each vertex is
initially labeled $0$ or $1$, there exists an $r$-local proof labeling
scheme of size $s$ recognizing all graphs of $\cG$ that contain
exactly one vertex $v$ with $I(v) = 1$. In general, we do not
explicitly refer to the initial labeling $I$ and we say that the proof
labeling scheme can distinguish a leader from the other vertices.

For both recognition and leader election, we design proof labeling
schemes where the certificate given to each vertex is either an integer
in $\N$, or an integer in $\{0,1,2\}$. We provide local verifiers that
ensure that in each graph $G$ endowed with such a numbering, each
vertex $v$ is labeled either by $d(s,v)$ or by $f_s(v) = d(s,v) \pmod{3}$,
where $d(s,v)$ is its distance to a root $s \in V(G)$.  We say that a
class $\cG$ admits an \emph{$r$-local distance certification protocol}
if there exists an $r$-local verifier that accepts at every vertex $v$
in a graph $G \in \cG$ if and only if there exists a vertex $s$ such
that each vertex $v \in G$ is labeled by $d(s,v)$. Similarly, the
class $\cG$ admits an \emph{$r$-local distance mod 3 certification
  protocol} if there exists an $r$-local verifier that accepts at every
vertex $v$ in a graph $G \in \cG$ if and only if there exists a vertex
$s$ such that each vertex $v \in G$ is labeled by $d(s,v)
\pmod{3}$. 

Observe that for both kinds of local certification protocols, if a
graph $G$ is accepted, then the root $s$ is necessarily unique and can
be recognized locally. Indeed, in the case of a local distance
certification protocol, $s$ is the only vertex labeled by $0$. In 
the case of a local distance mod 3 certification protocol, $s$ is the
only vertex $v$ that does not have a neighbor $u$ such that
$d(s,u) = d(s,v) -1 \pmod{3}$. Consequently, one can obtain an
$r$-local leader election protocol of size $O(\log D)$ (respectively,
$O(1)$) from a local distance (respectively, local distance mod $3$)
certification protocol by checking that the root $s$ is the only
vertex $v$ with $I(v)=1$.

\section{Meshed and weakly modular graphs}\label{sec-wm} To define meshed and weakly modular graphs, first we introduce three
conditions, which are defined globally but concern distances from a
vertex $s$ to two vertices $u,v$ at distance 1 or 2 (see
Figure~\ref{fig-TC-QC}):

\begin{enumerate}
\item[(TC)] Triangle Condition: for any  $u,v,s \in V(G)$ with
  $d(u,v)=1$ and $d(u,s)=d(v,s)$, there exists a common neighbor $x$
  of $u$ and $v$ such that $d(s,x)=d(s,u)-1$.
\item[(QC)] Quadrangle Condition: for any $u,v,w,s\in V(G)$ with
  $d(u,w)=d(v,w)=1$, $d(u,v) = 2$, and $d(s,u)=d(s,v)=d(s,w)-1,$ there
  exists a common neighbor $x$ of $u$ and $v$ such that
  $d(s,x)=d(s,u)-1$.
\item[(QC$^-$)] Weak Quadrangle Condition: for any  $u,v,s\in V(G)$ with
$d(u,v)=2$, there exists a common neighbor $x$ of $u$ and $v$ such that $2d(s,x)\le d(s,u)+d(s,v).$
\end{enumerate}

If $s$ is fixed and (TC), (QC), or (QC$^-$) holds for any vertices
$u,v$, then we refer to such condition as (TC$(s)$), (QC$(s)$), or
(QC$^-(s)$). Since if $x\sim u,v$, $d(s,x)$ differs by at most 1 from
each $d(s,u)$ and $d(s,v)$, (QC$^-$) is equivalent to the two
following conditions (see Figure~\ref{fig-TC-QC}).
\begin{enumerate}
\item[(QC$_1^-$)] for any $u,v,w,s\in V(G)$ with $d(u,w)=d(v,w)=1$,
  $d(u,v)=2$ and $d(s,u) = d(s,w) = d(s,v) +1$, there exists a common
  neighbor $x$ of $u$ and $v$ such that $d(s,x) = d(s,v)$.
\item[(QC$_2^-$)] for any $u,v,w,s\in V(G)$ with $d(u,w)=d(v,w)=1$,
  $d(u,v) = 2$, and $d(s,u)=d(s,v)=d(s,w)-1,$ there exists a common
  neighbor $x$ of $u$ and $v$ such that
  $d(s,u)-1 \leq d(s,x) \leq d(s,u)$.
\end{enumerate}

\begin{figure}[ht]
\begin{center}
\includegraphics{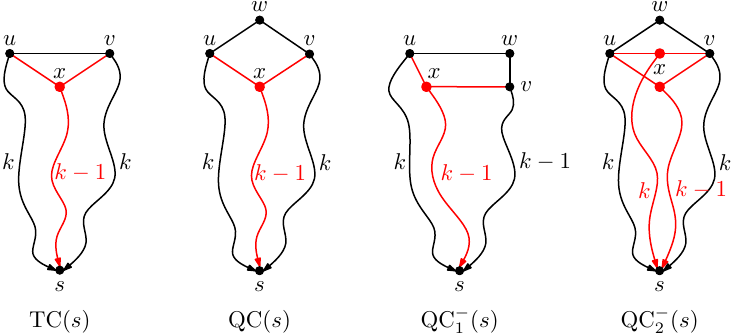}
\end{center}
\caption{The triangle, quadrangle and weak quadrangle
  conditions}\label{fig-TC-QC}
\end{figure}

\begin{definition} [Meshed and weakly modular graphs] 
A graph $G$ is  \emph{meshed}  if $G$ satisfies  (QC$^-$) and $G$ is  
\emph{weakly modular} if $G$ satisfies (TC) and (QC).  Analogously, $G$ is  \emph{meshed with respect to $s$}   if $G$ satisfies  (QC$^-(s)$) and \emph{weakly modular with respect to $s$} if $G$ satisfies (TC$(s)$) and (QC$(s)$).
\end{definition}

Weakly modular graphs have been introduced in \cite{BaCh_helly,Ch_metric} with the purpose of giving a common generalization of many classes of graphs defined by distance properties. Meshed graphs have been introduced in
the unpublished paper \cite{BaMuSo} as a generalization of weakly modular graphs. Indeed, (TC) and (QC) imply
(QC$^-$) and thus weakly modular graphs are meshed. On the other hand, the weak quadrangle condition (QC$^-$) implies the triangle condition (TC). Weakly modular graphs and their subclasses have been studied in numerous papers from metric graph theory and geometric group theory, see for example the survey \cite{BaCh_survey}, the monograph \cite{CCHO}, and some references below. Meshed graphs comprise several important classes of graphs (e.g., basis graphs of matroids) which are not weakly modular. They have been further studied in the papers \cite{BaCh_median,CCHO,Ch_delta,Ch_S3}. We continue with some properties of these classes and with the definition of their subclasses.

Three vertices $v_1,v_2,v_3$ of a graph $G$ form a {\it metric triangle}
$v_1v_2v_3$ if the intervals $I(v_1,v_2), I(v_2,v_3),$ and
$I(v_3,v_1)$ pairwise intersect only in the common end-vertices, i.e.,
$I(v_i, v_j) \cap I(v_i,v_k)= \{v_i\}$ for any $1 \leq i, j, k \leq 3$.
If $d(v_1,v_2)=d(v_2,v_3)=d(v_3,v_1)=k,$ then this metric triangle is
called {\it equilateral} of {\it size} $k.$ An equilateral metric
triangle $v_1v_2v_3$ of size $k$ is called \emph{strongly equilateral}
if $d(v_1,v)=k$ for all $v\in I(v_2,v_3)$. Weakly modular graphs are
characterized by their metric triangles.

\begin{lemma}[\!\!\cite{Ch_metric}]\label{strongly-equilateral}   A graph $G$ is weakly modular
  if and only if all metric triangles of $G$ are strongly equilateral.
\end{lemma}

A \emph{median} $m$ of a triplet of vertices $x,y,z$ is a vertex that
lies simultaneously on shortest paths from $x$ to $y$, from $x$ to
$z$, and from $y$ to $z$, i.e.,
$m \in I(x,y) \cap I(x,z) \cap I(y,z)$. For a triplet of vertices, a
median may not exist or may not be unique.

The following classes of graphs are subclasses of weakly modular
graphs (the proof follows in most cases from Lemma
\ref{strongly-equilateral}):
\begin{itemize}
\item{\emph{Modular graphs}:} Each triplet of vertices $x,y,z$ has a median,
  i.e., each metric triangle has size 0;
\item{\emph{Median graphs}:} Each triplet of vertices $x,y,z$ has a unique median (trees, hypercubes and multidimensional grids are simple examples of median graphs);
\item{\emph{Pseudo-modular graphs}~\cite{BaMu}:}  Each metric triangle has size 0 or 1; 
\item{\emph{Helly graphs}}~\cite{BP-absolute,BaPr}: The balls satisfy
  the Helly property: any family of pairwise intersecting balls has a
  non-empty intersection (an example of a Helly graph is the king grid
  ${\mathbb Z}_{\infty}^d$ obtained from ${\mathbb Z}^d$ by replacing
  each cube by a clique) ;
\item{\emph{Bridged graphs}}~\cite{FaJa,SoCh}: Balls around convex
  sets are convex and equivalently all isometric cycles have length 3
  (chordal graphs are examples of bridged graphs);
\item{\emph{Weakly bridged graphs}}~\cite{ChOs}: Weakly modular graphs
  with convex balls;
\item{\emph{Sweakly modular graphs}}~\cite{CCHO}: Weakly modular
  graphs without $K_4^-$ and $K^-_{3,3}$ ($K_4^-$ and $K^-_{3,3}$ are
  the graphs obtained respectively by removing an edge from the
  complete graph $K_4$ and the complete bipartite graph $K_{3,3}$);
\item{\emph{Dual polar graphs}}~\cite{Ca}: The incidence graphs of
  maximal subspaces of polar spaces and equivalently thick weakly
  modular graphs not containing $K_4^-$ and $K^-_{3,3}$ (a graph is
  \emph{thick} if every pair of vertices $u,v$ at distance $2$ have two
  non-adjacent common neighbors);
\item{\emph{Bucolic graphs}}~\cite{BrChChGoOs}: Weakly modular graphs not
  containing $K_{2,3}, W_4$ and $W_4^-$ (the wheel $W_k$ is obtained
  by adding a universal vertex to the cycle $C_k$, and $W_k^-$ is
  obtained by removing an edge from $W_k$);
\item{\emph{Cage-amalgamation graphs}} \cite{BrChChKoLaVa}: Weakly
  modular graphs not containing $K_{2,3}, W_4^-$, and $W_k$ for all
  $k\ge 4$.
\end{itemize}

In the mentioned papers and in other papers, these graphs and their
local and global structures have been thoroughly investigated (such
results will be presented and used in Section~\ref{sec-recognition}).

In meshed graphs, metric triangles are equilateral~\cite{BaCh_median}
but not always strongly equilateral. The graphs in the following
classes are meshed but not always weakly modular~\cite{Ch_delta}.
\begin{itemize}
\item{\emph{Basis graphs of matroids} \cite{Mau}:} The graphs whose vertices are the bases of a matroid and two bases are adjacent if one can be obtained from the second by a single exchange;
\item{\emph{Basis graphs of even $\Delta$-matroids} \cite{Bou,ChKa,DrHa}:} The graphs whose vertices are the bases of an even $\Delta$-matroid and two bases are adjacent if one can be obtained from the second by a single exchange.
\end{itemize}

\section{Local certification of distances}

In this section, we consider labelings with labels within $\N$ and we
want to define a $2$-local distance certification protocol for meshed
graphs, i.e., we want to ensure that for any meshed graph $G$, a
labeling $L: V(G) \to \N$ is accepted by the $2$-local verifier if and
only if there exists a root $s$ such that $L(v) = d(s,v)$ for all
$v \in V(G)$.

The main idea of our certification protocol is to verify that
QC$^-(s)$ by checking that the condition holds around each vertex $q$ of
$G$  when we replace $d(v,s)$ by $D(v)$ at every vertex $v$. We
first describe rules that are satisfied in any graph $G$ when every
vertex $v$ is labeled by $d(s,v)$ for some arbitrary root $s$.

\begin{enumerate} 
\item[\rDGz:] If $D(v)=0$, then for all $u\in N(v)$, we have $D(u)=1$.
\item[\rDG:] If $D(v)>0$, then there exists $u \in N(v)$ with
  $D(u) = D(v) -1$ and for all $w\in N(v)$, we have
  $D(v) -1 \leq D(w) \leq D(v) + 1$.
\end{enumerate}

We now consider an extra rule that holds for $D = d(s,\cdot)$ if and only if
$G$ satisfies QC$^-(s)$.

\begin{enumerate}
\item[\rDM:] If $u,u'\in N(v), u\nsim u'$, and either $D(u)=D(v)=D(u')+1$ or
  $D(u)=D(u')=D(v)-1$, then there exists $x\sim u,u'$ such that
  $D(x)\le D(u')$.
\end{enumerate}

\begin{theorem}\label{th-meshed-distance}
  The rules \rDGz, \rDG, and \rDM define a 2-local distance
  certification protocol for the class of meshed graphs.
\end{theorem}

\begin{proof}
  Consider a meshed graph $G$.  We first show that for any root
  $s \in V(G)$, the labeling $D = d(s,\cdot): V(G) \to \N$ satisfies
  \rDGz, \rDG, and \rDM. Note that $s$ is the only vertex $v$ such
  that $d(s,v) = 0$ and $\rDGz$ trivially holds at $s$. For any
  $v \neq s$, \rDG holds at $v$ since for any $w \sim v$,
  $|d(s,v)-d(s,w) | \leq 1$, and since $v$ has a neighbor $u$ on a
  shortest path from $v$ to $s$ that satisfies $d(s,u) =
  d(s,v)-1$. The fact that \rDM holds follows from the fact that $G$
  satisfies QC$^-(s)$. Indeed, $u$ and $u'$ are at distance $2$ in $G$
  and if $D(u)= D(u')+1$, QC$^-(s)$ implies that there exists
  $x \sim u,u'$ with $d(s,x) =d(s,u')$. Consequently, $D(x)=D(u')$ in this case. 
  If $D(u) = D(u')$, QC$^-(s)$
  implies that there exists $x \sim u,u'$ with
  $d(s,x) \in \{d(s,u')-1,d(s,u')\}$, yielding  $D(x)\in \{D(u')-1,D(u')\}$.

  Suppose now that we are given a labeling function $D: V(G) \to \N$
  that satisfies \rDGz, \rDG, and \rDM. Consider a vertex $s \in V(G)$
  with a minimum value of $D$. We show by induction on $k = d(s,v)$
  that for any $v \in V(G)$, we have $D(v) = d(s,v)$. Note that by
  \rDG applied at $s$, we have $D(s) = 0 = d(s,s)$. Moreover, by
  \rDGz, $D(u) = 1 = d(s,u)$ for every $u \in N(s)$. Consequently, the
  claim holds for $k \leq 1$. Assume that the property holds for all
  vertices at distance at most $k \geq 1$ from $s$ and consider a
  vertex $v$ such that $d(s,v) = k+1 \geq 2$. Let $u \sim v$ such that
  $d(s,u) = k$ and let $w \sim u$ such that $d(s,w) = k-1$. Note that
  $v \nsim w$ as $d(s,v) = d(s,w)+2$. By induction hypothesis
  $D(u) = d(s,u) = k$ and $D(w) = d(s,w) = k-1$. By \rDG applied at
  $u$, we have that $k-1 \leq D(v) \leq k+1$. If $D(v) \leq k$, then
  by \rDM applied at $u$, there exists $x \sim w,v$ such that
  $D(x) \leq D(w)=k-1$. But since $x \sim v,w$, necessarily
  $d(s,x) =k$, and by induction hypothesis, $D(x) =k$, a
  contradiction. This shows that $D(v) =k+1 = d(s,w)$ and ends the
  proof that $D$ coincides with $d(s,\cdot)$.

  Finally, note that one can check whether the rules $\rDGz$ and
  $\rDG$ hold at $v$ by considering the labeled 1-ball $B_1(v)$
  centered at $v$, while one can check whether $\rDM$ holds at $v$ by
  considering the labeled 2-ball $B_2(v)$. This shows that \rDGz,
  \rDG, and \rDM define a 2-local distance certification protocol for
  the class of meshed graphs.
\end{proof}

In the case of (weakly) bridged and Helly graphs, one can find 1-local
distance certification protocols by replacing \rDM by other
verification rules. In the case of bridged graphs, we consider the two
following rules. The first one corresponds to the triangle condition
(TC($s$)) while the second one encodes the local convexity of the
balls centered at $s$ (this property holds for weakly bridged graphs
and then also for bridged and chordal graphs).

\begin{enumerate}
\item[\rDTC:] For any $u \in N(v)$ such that $D(u)=D(v)$, 
  there exists $x\sim u,v$ with $D(x)=D(v)-1$.
\item[\rDB:] For any $u,u'\in N(v)$ such that
  $D(u)=D(u')=D(v)-1$, we have $u\sim u'$.
\end{enumerate}

\begin{theorem}\label{th-bridged-distance}
  The rules \rDGz, \rDG, \rDTC and \rDB define a 1-local distance
  certification protocol for the class of weakly bridged graphs.
\end{theorem}

\begin{proof}
  Consider a weakly bridged graph $G$.  We first show that for any
  root $s \in V(G)$, the labeling $D = d(s,\cdot): V(G) \to \N$
  satisfies \rDGz, \rDG, \rDTC, and \rDB. \rDGz and \rDG hold for the
  same reasons as in the proof of Theorem~\ref{th-meshed-distance}.
  Since weakly bridged graphs are weakly modular, they satisfy the
  triangle condition. Consequently, for any $v$ and any $u \sim v$
  such that $d(s,u) = d(s,v)$, there exists $x \sim u,v$ such that
  $d(s,x)=d(s,v)-1$ and \rDTC holds at $v$. Consider now a vertex $v$
  and two neighbors $u,u'$ of $v$ such that
  $d(s,u)=d(s,u')=d(s,v)-1$. Let $k = d(s,u)$ and observe that
  $u,u' \in B_k(s)$ while $v \notin B_k(s)$. Since weakly bridged
  graphs have convex balls, $v$ cannot belong to a shortest path from
  $u$ to $u'$ and then $u$ and $u'$ are at distance $1$. This shows
  that \rDB holds at any $v$.

  Suppose now that we are given a labeling function $D: V(G) \to \N$
  that satisfies \rDGz, \rDG, \rDTC, and \rDB.  Consider a vertex
  $s \in V(G)$ with a minimum value of $D$. As in the proof of
  Theorem~\ref{th-meshed-distance}, we show by induction on
  $k = d(s,v)$ that for any $v \in V(G)$, we have $D(v) = d(s,v)$. For
  $k \leq 1$, the proof is the same as in the proof of
  Theorem~\ref{th-meshed-distance}. Assume now that the property holds
  for all vertices at distance at most $k \geq 1$ from $s$ and
  consider a vertex $v$ such that $d(s,v) = k+1 \geq 2$. Let $u \sim v$ such
  that $d(s,u) = k$ and let $w \sim u$ such that $d(s,w) = k-1$. Note
  that $v \nsim w$ as $d(s,v) = d(s,w)+2$.  By induction hypothesis
  $D(u) = d(s,u) = k$ and $D(w) = d(s,w) = k-1$. By \rDG applied at
  $u$, we have that $k-1 \leq D(v) \leq k+1$. If $D(v) = k-1$, then by
  \rDB applied at $u$, $w \sim v$, a contradiction. If $D(v) = k$,
  then by \rDTC applied at $u$, there exists $x \sim u,v$ such that
  $D(x) = k-1$. If $d(s,x) = k+1$, we have a contradiction by the
  previous case replacing $v$ by $x$. Consequently, $d(s,x) \leq k$
  and thus by induction hypothesis, $d(s,x) = D(x) = k-1$. But this is
  impossible as $x$ is a neighbor of $v$ and $d(s,v) = k+1$. This
  shows that $D(v) = k+1$ and ends the proof that $D$ coincides with
  $d(s,\cdot)$.

  One can easily check that \rDGz, \rDG, \rDTC, and \rDB define a
  1-local verifier.
\end{proof}

In the case of Helly graphs, we consider the following rule that
corresponds to a domination property satisfied in Helly graphs that
was established by Bandelt and Pesch~\cite{BP-absolute}.
\begin{enumerate}
\item[\rDH:] If $D(v)>0$, then there exists $x \in N(v)$ with
  $D(x)=D(v)-1$ such that for any $u \in N(v)$ such that
  $D(u)\le D(v)$, we have $u\in N[x]$.
\end{enumerate}

\begin{theorem}\label{th-helly-distance}
  The rules \rDGz, \rDG, and \rDH define a 1-local distance
  certification protocol for the class of Helly graphs.
\end{theorem}

\begin{proof}
  Consider a Helly graph $G$.  We first show that for any root
  $s \in V(G)$, the labeling $D = d(s,\cdot): V(G) \to \N$ satisfies
  \rDGz, \rDG, and \rDH. \rDGz and \rDG hold for the same reasons as
  in the proof of Theorem~\ref{th-meshed-distance}. Consider a vertex
  $v$ such that $D(v) = d(s,v) = k+1$ and let $S$ be the set of all
  neighbors $u$ of $v$ such that $d(s,u) \leq k+1$. Consider the
  following collection of balls: $B_k(s)$, $B_1(v)$, and $B_1(u)$ for
  each $u \in S$. The balls of this collection pairwise intersect, and
  consequently, by the Helly property, there exists a vertex $x$ in
  the intersection of all these balls. Since $x$ belongs to $B_k(s)$
  and $x \sim v$, we have $D(x)=d(s,x)=k$. Since $x$ belongs to
  $B_1(v)$ and to $B_1(u)$ for every $u\in S$,
  $S\cup \{v \}\subseteq N[x]$, establishing \rDH.

  Suppose now that we are given a labeling function $D: V(G) \to \N$
  that satisfies \rDGz, \rDG, and \rDH.  Consider a vertex
  $s \in V(G)$ with a minimum value of $D$. As in the proof of
  Theorem~\ref{th-meshed-distance}, we show by induction on
  $k = d(s,v)$ that for any $v \in V(G)$, we have $D(v) = d(s,v)$. For
  $k \leq 1$, the proof is the same as in the proof of
  Theorem~\ref{th-meshed-distance}. Assume now that the property holds
  for all vertices at distance at most $k \geq 1$ from $s$ and
  consider a vertex $v$ such that $d(s,v) = k+1 \geq 2$. Let
  $u \sim v$ such that $d(s,u) = k$ and let $w \sim u$ such that
  $d(s,w) = k-1$. Note that $v \nsim w$ as $d(s,v) = d(s,w)+2$.  By
  induction hypothesis $D(u) = d(s,u) = k$ and $D(w) = d(s,w) =
  k-1$. By \rDG applied at $u$, we have that $k-1 \leq D(v) \leq
  k+1$. If $D(v) \leq D(u) = k$, then by \rDH applied at $u$, there
  exists $x \sim u$ such that $D(x) = k-1$ and $x \sim
  w,v$. Consequently, $d(s,x) = k$ and by induction hypothesis, we
  should have $D(s,x) = d(s,x) = k$, a contradiction. This shows that
  $D(v) = k+1$ and ends the proof that $D$ coincides with
  $d(s,\cdot)$.

  One can easily check that \rDGz, \rDG, and \rDH define a 1-local
  verifier.
\end{proof}

\section{Leader election  by counting distances modulo 3}\label{sec-election}

In this section, we consider labelings with labels within $\{0,1,2\}$
and we want to define  $r$-local distance mod 3 certification
protocols, i.e., we want to ensure that for any graph $G$ in the
considered family $\cG$, a labeling $L: V(G) \to \{0,1,2\}$ is
accepted by the $r$-local verifier if and only if there exists a root
$s$ such that $L(v) = d(s,v) \pmod{3}$ at every vertex $v$. Given a
root $s$, let $f_s: V(G) \to \{0,1,2\}$ be the map defined by
$v \mapsto d(s,v) \pmod{3}$.

For any $k \in \{0,1,2\}$, let $\precc(k) = k-1 \pmod{3}$ and
$\next(k) = k+1 \pmod{3}$. Given a graph $G$ and a labeling
$L:V(G) \rightarrow \{ 0,1,2\}$, we define the following directed
graph $\oG_L$: the vertex-set of $\oG_L$ is $V(G)$ and there is an arc
$\overrightarrow{vu}$ in $\oG_L$ if and only if $vu$ is an edge of $G$
and $L(v)=\next(L(u))$. One can observe that if $L=f_s$ for some root
$s$, then $\oG_L$ is acyclic and $s$ is the unique sink of $\oG_L$.
In order to prove our results, we show that our local rules ensure
that for any meshed graph $G$ and any labeling
$L:V(G) \rightarrow \{ 0,1,2\}$, whenever all rules are satisfied, the
directed graph $\oG_L$ is acyclic and admits a unique sink.

We first describe rules that are satisfied at any vertex $v$ in any
graph $G$ when $L = f_s$ for some arbitrary root $s$ (even if $G$ is
not meshed). The first rule can be seen as a modulo $3$ version of
rules $\rDGz$ and $\rDG$, while $\rMGt$ and $\rMGs$ ensure that the
labeling is coherent on every triangle and square of the graph.

\begin{enumerate} 
\item[\rMG:] Either there exists $u \in N(v)$ such that
  $L(u) = \precc(L(v))$, or $L(v) = 0$ and $L(u) = 1$ for any
  $u \in N(v)$.
\item[\rMGt:] For any triangle $uvw$ of $G$, we have
  $\{ L(u),L(v),L(w)\}\ne \{ 0,1,2\}$.
\item[\rMGs:] For any square  $uvwx$ of $G$ such that
  $\{ L(u),L(v),L(w),L(x)\}=\{ 0,1,2\}$, the vertices of $uvwx$
  with the same label are non-adjacent.\end{enumerate}
Observe that one can verify if \rMG and \rMGt are satisfied at a
vertex $v$ by considering $B_1(v)$, and that \rMGs can be checked by
considering $B_2(v)$.  We now consider an extra rule for meshed graphs
that can also be checked by considering $B_2(v)$. This
rule can be seen as a modulo $3$ version of rule $\rDM$.

\begin{enumerate} 
\item[\rMM:] If $u,u'\in N(v), u \nsim u'$, and either $L(u)=L(v)=\next(L(u'))$ or $L(u)=L(u')=\precc(L(v))$, then there exists $x\sim u,u'$ such that $L(x)=L(u')$ or $L(x)=\precc(L(u'))$.
\end{enumerate}

\begin{theorem}\label{meshed-modulo3}
  The rules \rMG, \rMGs, \rMGt and \rMM define a 2-local distance
  mod 3 certification protocol for the class of meshed graphs.

  There exists a $2$-local leader election protocol for
  meshed graphs using labels $\{0,1,2\}$.
\end{theorem}

In the case of weakly bridged and Helly graphs, we show that we can
find $1$-local distance mod 3 certification protocols by replacing
\rMGs and \rMM by rules that can be verified at a vertex $v$ by
considering only the labeled $1$-ball $B_1(v)$.
For weakly bridged graphs, we consider the two following rules that
correspond to the modulo 3 versions of rules \rDTC and \rDB:
\begin{enumerate}
\item[\rMTC:] If $u\in N(v)$, $L(u)=L(v)$, then there exist $x\sim u,v$ such that $L(x)=\precc(L(v))$.
\item[\rMB:] If $u,u'\in N(v)$ and $L(u)=L(u')=\precc(L(v))$, then $u\sim u'$. \end{enumerate}

Finally, for Helly graphs, we consider the following rule, that is the modulo 3 version of \rDH:
\begin{enumerate}
\item[\rMH:] If there exists $y\in N(v)$ with $L(y)=\precc(L(v))$, then there exists $x\sim v$ with $L(x)=\precc(L(v))$ such that if $u\sim v$ and $L(u)=L(v)$ or $L(u)=\precc(L(v))$, then $u\in N[x]$.
\end{enumerate}

\begin{theorem}\label{th-bridged-helly-election}
  The rules \rMG, \rMGt, \rMTC and \rMB define a 1-local distance
  mod 3 certification protocol for the class of weakly bridged
  graphs. Therefore, there exists a $1$-local leader election protocol
  for weakly bridged graphs using labels $\{0,1,2\}$.

  The rules \rMG, \rMGt, and \rMH define a 1-local distance mod 3
  certification protocol for the class of Helly graphs. Therefore,
  there exists a $1$-local leader election protocol in Helly graphs
  using labels $\{0,1,2\}$.
\end{theorem}

\subsection{Proof of Theorem~\ref{meshed-modulo3}}

We first show that in any meshed graph $G$, for any root $s$, the
labeling $f_s:V(G) \to \{0,1,2\}$ satisfies rules \rMG, \rMGt, \rMGs,
\rMM. We first establish that the three rules are satisfied in any graph $G$ (even if $G$ is not meshed). 

\begin{lemma}\label{lem-gen-fr-ok}
  In any graph $G$, for any root $s \in V(G)$, the labeling
  $f_s: v \in V(G) \mapsto d(s,v) \pmod{3}$ satisfies \rMG, \rMGt,
  \rMGs.
\end{lemma}

\begin{proof}
  Consider a vertex $v \in V(G)$. If $v = s$, then $d(s,v) = 0$ and
  $d(s,v) = 1$ for all $u \in N(v)$. Otherwise, there exists
  $u \in N(v)$ such that $d(s,u) = d(s,v) -1$ and then
  $f_s(u) = \precc(f_s(v))$. Consequently, \rMG holds for every vertex
  $v \in V(G)$.

  Note that for any edge $uv$ of $G$,
  $d(s,v)-1 \leq d(s,u) \leq d(s,v)+1$. Consequently, if
  $f_s(u) = \precc(f_s(v))$ (respectively, $f_s(u) = f_s(v)$, or
  $f_s(u) = \next(f_s(v))$), then $d(s,u) = d(s,v) -1$ (respectively,
  $d(s,u) = d(s,v)$, or $d(s,u) = d(s,v) +1$). Suppose that \rMGt does
  not hold and consider a triangle $uvw$ of $G$ such that
  $L(u) = 0, L(v) = 1, L(w) = 2$. Then we have
  $d(s,u) = d(s,v) -1 = d(s,w) -2$, but this is impossible since $uw$ is
  an edge of $G$.  Consequently, \rMGt holds on every triangle of
  $G$. Suppose now that there exists a square $uvwx$ on which $\rMGt$
  does not hold. Without loss of generality, assume that
  $L(u) = 0, L(v) = 1, L(w) = 1$, and $L(x) = 2$. Then
  $d(s,u) = d(s,v)-1 = d(s,w) -1 = d(s,x)-2$. Again, this is impossible
  as $ux$ is an edge of $G$, and consequently, \rMGs holds on every
  square of $G$.
\end{proof}

We now show that $f_s$ satisfies $\rMM$ when the graph is meshed.
\begin{lemma}\label{lem-meshed-fr-ok}
  For any graph $G$ and any root $s \in V(G)$ such that $G$ satisfies
  QC$^-(s)$, the labeling $f_s: v \in V(G) \mapsto d(s,v) \pmod{3}$
  satisfies \rMM.
\end{lemma}

\begin{proof}
  Consider a vertex $v \in V(G)$ and two non-adjacent neighbors $u,u'$
  of $v$ in $G$. Suppose first that
  $f_s(u)=f_s(v)=\next(f_s(u'))$. This implies that
  $d(s,u) = d(s,v) = d(s,u') +1$. By (QC$_1^-(s)$), there exists
  $x \sim u,u'$ such that $d(s,x) = d(s,u') = d(s,u)-1$, and thus
  $f_s(x) = f_s(u')$.  Suppose now that
  $f_s(u)=f_s(u')=\precc(f_s(v))$. This implies that
  $d(s,u) = d(s,u') = d(s,v) -1$. By (QC$_2^-(s)$), there exists
  $x \sim u,u'$ such that $d(s,u') - 1 \leq d(s,x) \leq
  d(s,u')$. Consequently, $L(x) \in \{L(u'),\precc(L(u'))\}$. In both
  cases, there exists $x \sim u, u'$ with
  $L(x) \in \{L(u'),\precc(L(u'))\}$.  This shows that \rMM holds at
  every vertex $v \in V(G)$.
\end{proof}

We now consider a meshed graph $G$ and a labeling
$L: V(G)\rightarrow \{ 0,1,2\}$ satisfying rules \rMG, \rMGt, \rMGs,
and \rMM. We show that there exists a root $s \in V(G)$ such that
$L = f_s$. To do so, we consider the directed graph $\oG_L$ and we
want to show that $\oG_L$ is acyclic and contains a unique sink that
coincide with the root $s$.

We first show that if  the triangle-square complex of a graph $G$ is simply connected
(which is the case if $G$ is meshed~\cite{CCHO}), then \rMGt and \rMGs are enough to ensure the acyclicity
of $\oG_L$.

\begin{lemma}\label{lem-acyclic}
  If the triangle complex $X\tr(G)$ (respectively, the triangle-square
  complex $X\trsq(G)$) of a graph $G$ is simply connected and
  $L: V(G)\rightarrow \{0,1,2\}$ is a labeling satisfying the rule
  \rMGt (respectively, the rules \rMGt and \rMGs), then $\oG_L$ is
  acyclic.
\end{lemma}

The proof of \Cref{lem-acyclic} relies on a generalization of
Sperner's Lemma that is due to Musin~\cite{Musin}. Let $T$ be a plane
triangulation and let $\partial T = (u_0, u_1, \ldots, u_{m})$ be its
outer-face (or boundary). Consider an arbitrary labeling
$L: T\rightarrow \{ 0,1,2\}$ of the vertices of $T$ with labels
$0,1,2$.  A triangle $uvw$ of $T$ is a \emph{fully colored triangle}
if $\{ L(u),L(v),L(w)\}=\{ 0,1,2\}$.  Let $p_{01}$ (respectively, $p_{10}$)
denote the number of pairs of consecutive vertices $u_i$ and $u_{i+1}$
of $\partial T$ such that $L(u_i) = 0$ and $L(u_{i+1}) = 1$
(respectively $L(u_i) = 1$ and $L(u_{i+1}) = 0$), and let
$\deg([0,1],T,L)=p_{01}-p_{10}$. The quantities $\deg([1,2],T,L)$ and
$\deg([2,0],T,L)$ are defined in a similar way and it is shown in
\cite[Lemma 2.1]{Musin} that
$\deg([0,1],T,L)=\deg([1,2],T,L)=\deg([2,0],T,L)$. Musin established
the following generalization of Sperner's Lemma.

\begin{proposition}[{\!\!\cite[Corollary 3.1]{Musin}}]\label{th-musin}
  For any planar triangulation $T$, for any labeling
  $L: T\rightarrow \{ 0,1,2\}$, $T$ contains at least
  $|\deg([0,1],T,L)|$ fully colored triangles.
\end{proposition}

\begin{proof}[Proof of \Cref{lem-acyclic}]
  We first prove the result when $X\tr(G)$ is simply connected and
  when the labeling $L$ satisfies \rMGt.  Suppose that $\oG_L$ is not
  acyclic and consider a directed cycle
  $C = (v_0, v_1, \dots, v_{\ell-1})$ in $\oG_L$. Note that $\ell=3k$
  is a multiple of $3$.  Since $X\tr(G)$ is simply connected, the
  cycle $C$ admits a \emph{disk diagram} $(D,\varphi)$ where $D$ is a plane
  triangulation and $\varphi: D\to X\tr(G)$ is a cellular map that
  maps bijectively the outer-face
  $\partial D = (u_0, u_1, \ldots, u_{\ell-1})$ of $D$ to
  $C = (v_0, v_1, \dots, v_{\ell-1})$ with $\varphi(u_i) = v_i$. For
  each vertex $u \in V(D)$, let $L'(u) = L(\varphi(v))$. Observe that
  for each $0 \leq i < \ell$,
  $L'(u_{i+1}) = L(v_{i+1}) = \precc(L(v_i)) = \precc(L'(u_i))$ and
  thus $|\deg([0,1],D,L')| = k \geq 1$. By \Cref{th-musin}, there
  exists a fully colored triangle $uvw$ in $D$. Since $u, v, w$ have
  different colors in $D$, they have different images via $\varphi$ in
  $V(G)$ and since $\varphi$ is a cellular map,
  $\varphi(u)\varphi(v)\varphi(w)$ is a fully colored triangle of $G$
  but this is impossible, as \rMGt holds in $G$.

  Consider now a graph $G$ with a simply connected triangle-square
  complex $X\trsq(G)$ and let $L: V\rightarrow \{ 0,1,2\}$ be a
  labeling satisfying \rMGt and \rMGs. We transform $G$ into a
  graph $G'$ with a simply connected triangle complex $X\tr(G')$
  as follows. For each square $S=uvwx$ of $G$ we add
  a new vertex $v_S$ and make $v_S$ adjacent to the vertices
  $u,v,w,x$, only (i.e., add the edges $v_Su,v_Sv,v_Sw$, and $v_Sx$ to
  $E'$). The triangle complex $X\tr(G')$ of $G'$ is obtained by adding
  to $X\tr(G)$ the triangles $v_Suv$, $v_Svw$, $v_Swx$, $v_Sxu$ for
  each square $S=uvwx$ of $G$.  In fact, $X\tr(G')$ can be viewed as
  the partial barycentric subdivision of $X\trsq(G)$ where each square
  is subdivided in four triangles. Since $X\trsq(G)$ is simply
  connected, $X\tr(G')$ is simply connected as well.

  Consider a labeling $L$ of $G$ satisfying \rMGt and \rMGs. We now
  define a labeling $L' : V(G') \to \{0,1,2\}$ satisfying \rMGt as
  follows. For each $v \in V(G)$, let $L'(v) = L(v)$ and for each
  square $S=uvwx$, let $L'(v_S)$ be any label that appears at least
  twice among $L(u),L(v),L(w),L(x)$. Clearly, $L'$ is a labeling from
  $V(G')$ to $\{0,1,2\}$. We now show that $L'$ satisfies
  \rMGt. Suppose this is not the case and consider a triangle $abc$ in
  $X\tr(G')$ such that $\{L(a),L(b),L(c)\} = \{0,1,2\}$. Since \rMGt
  holds in $G$, we can assume that there exists a square $S=uvwx$ such
  that $a = u$, $b= v$, and $c = v_S$. Since $L'(v_S)$ is a color that
  appears at least twice among $L(u),L(v),L(w),L(x)$ and since
  $L'(v_S) \notin \{L(u),L(v)\}$, we have that
  $L'(v_S) = L(w) = L(x)$. Consequently, in the square $S = uvwx$ of
  $G$, we have $\{L(u),L(v),L(w),L(x)\} = \{0,1,2\}$ and $S$ contains
  two adjacent vertices $w$ and $x$ with the same label. But this is
  impossible, as \rMGs holds in $G$. This shows that in $G'$, the
  labeling $L'$ satisfies \rMGt.  Since $G'$ has a simply connected
  triangle complex $X\tr(G')$, by what is shown above, we deduce that
  the directed graph $\oG'_{L'}$ is acyclic. Since $L(v)=L'(v)$ for
  any vertex $v$ of $G$ and since $G$ is a subgraph of $G'$, $\oG_L$
  is a subgraph of $\oG'_{L'}$, and thus $\oG_L$ is also acyclic. This
  concludes the proof of \Cref{lem-acyclic}.
\end{proof}

By \Cref{lem-acyclic}, we deduce that $\oG_L$ is acyclic and thus it
contains at least one source $s$. We now show that if \rMG and
\rMM hold, then $L$ coincides with the map $f_s$. The proof uses
similar ideas as the proof of Theorem~\ref{th-meshed-distance}.

\begin{lemma}\label{lem-meshed-mod}
  Consider a graph $G$ and a labeling $L: V\rightarrow \{ 0,1,2\}$
  satisfying \rMG and \rMM. For any sink $s$ of $\oG_L$ and any
  $v \in V(G)$, we have $L(v) = f_s(v) = d(s,v) \pmod{3}$.
\end{lemma}

\begin{proof}
  We prove the lemma by induction on $k=d(s,v)$. Since $s$ is a sink,
  $s$ has no neighbor labeled by $\precc(L(s))$ and thus by \rMG, we
  necessarily have $L(s) = 0$ and $L(u) = 1$ for any $u \in
  N(s)$. This shows that the property holds for $k =0$ and $k=1$.

  Assume that the assertion holds for all vertices at distance at most
  $k\geq 1 $ from $s$ and consider $v\in V(G)$ such that
  $d(s,v)=k+1\ge 2$. Pick any $u\sim v$ with $d(s,u)=k$ and any
  $w\sim u$ with $d(s,w)=k-1$. By induction hypothesis,
  $L(u) = d(s,u) \pmod{3}$ and $L(w)=d(s,w) \pmod{3}$. This implies
  that $L(w)=\precc(L(u))$. By contradiction, assume that
  $L(v) \neq d(s,w) \pmod{3}$. This implies that either
  $L(v)=\precc(L(u))=L(w)$ or $L(v)=L(u)=\next(L(w))$. In both cases,
  by \rMM applied to $u$, there exists $x\sim v,w$ such that
  $L(x) \in \{L(w),\precc(L(w))\}$.  Since $x\sim v,w$, we have
  $d(s,x) = k$ and thus by induction hypothesis, we have
  $L(x) = \next(L(w))$ which contradicts the fact that
  $L(x) \in \{L(w),\precc(L(w))\}$. This concludes the proof that
  $L(v)=d(s,v)\pmod{3}$.
\end{proof}

We can now prove Theorem~\ref{meshed-modulo3}

\begin{proof}[Proof of Theorem~\ref{meshed-modulo3}]
  The fact that in any meshed graph $G$, for any root $s$, the map
  $f_s$ satisfies \rMG, \rMGt, \rMGs, \rMM follows from
  Lemmas~\ref{lem-gen-fr-ok} and~\ref{lem-meshed-fr-ok}.

  Consider now a meshed graph $G$ and a labeling
  $L: V(G) \to \{0,1,2\}$ satisfying the rules  $\rMG, \rMGt, \rMGs, \rMM$
  at every $v \in V(G)$. Since the triangle-square complex $X\tr(G)$
  is simply connected~\cite[{Theorem~9.1}]{CCHO}, by
  \Cref{lem-acyclic}, $\oG_L$ is acyclic and thus contains at least
  one sink. By Lemma~\ref{lem-meshed-mod}, for any sink $s$ of $\oG_L$
  and any $v \in V(G)$, $L(v) = d(s,v) \pmod{3} = f_s(v)$. Suppose
  that $\oG_L$ contains two sinks $s, s'$ and let $u$ be a neighbor of
  $s$ on a shortest path from $s$ to $s'$. By
  Lemma~\ref{lem-meshed-mod} applied to $s$, we have $L(s) = 0$ and
  $L(u) = 1$. By Lemma~\ref{lem-meshed-mod} applied to $s'$, we have
  $L(u) = \precc(L(s))$, which is a contradiction. This shows that
  $\oG_L$ has a unique sink.  Designating the sink as a leader and every
  non-sink vertex as a non-leader, we thus have designed a local
  leader election protocol for meshed graphs using only the three
  labels $0,1,2$.

  Note that these protocols are $2$-local. Indeed, for any
  $v \in V(G)$, we can check whether the rules $\rMG$ and $\rMGt$
  hold at $v$ by considering only the labeled ball $B_1(v)$, and
  whether the rules $\rMG$ and $\rMGs$ hold at $v \in V(G)$ by
  considering the labeled ball $B_2(v)$.
\end{proof}

\subsection{Proof of Theorem~\ref{th-bridged-helly-election}}

The proof of Theorem~\ref{th-bridged-helly-election} follows exactly
the same approach as the proof of \Cref{meshed-modulo3}. We first show
that for any root $s \in V(G)$, the map $f_s$ satisfies the local
rules \rMG, \rMGt, \rMTC, \rMB (respectively, \rMG, \rMGt, \rMGs,
\rMH) if $G$ is a weakly bridged graph (respectively, a Helly graph).

\begin{lemma}\label{lem-fr-bridged-ok}
  For any weakly bridged graph $G$ and any root $s \in V(G)$, the labeling
  $f_s: v \mapsto d(s,v) \pmod{3}$ satisfies \rMG, \rMGt, \rMTC, \rMB.
\end{lemma}

\begin{proof}
  The fact that $f_s$ satisfies \rMG and \rMGt follows from
  \Cref{lem-gen-fr-ok}. Consider two adjacent vertices $u,v$ such that
  $L(u) = L(v)$. This implies that $d(s,v) = d(s,u)$. By TC$(s)$,
  there exists $x \sim u,v$ such that $d(s,x) = d(s,v)
  -1$. Consequently, $L(x) = \precc(L(v))$ and \rMTC holds.
Consider now a vertex $v$ and two neighbors $u, u'$ of $v$ such that
  $L(u)=L(u')=\precc(L(v))$. This implies that
  $d(s,u) = d(s,u') = d(s,v)-1$. Since the balls in $G$ are convex,
  $v$ does not belong to a shortest path from $u$ to $u'$, and
  consequently, $u\sim u'$. This shows that \rMB holds. 
\end{proof}

\begin{lemma}\label{lem-fr-helly-ok}
  For any Helly graph $G$ and any root $s \in V(G)$, the labeling
  $f_s: v \mapsto d(s,v) \pmod{3}$ satisfies \rMG, \rMGt, \rMGs, \rMH.
\end{lemma}

\begin{proof}
  The fact that $f_s$ satisfies \rMG and \rMGs follows from
  \Cref{lem-gen-fr-ok}. Consider a vertex $v$, let $k = d(s,v)$, and
  assume that there exists a vertex $y \sim v$ such that
  $L(y) = \precc(L(v))$. Note that this implies that $L(y) = k-1$ and
  thus $d(s,v) > 0$. Since $d(s,\cdot)$ satisfies \rDH, there exists
  $x \sim v$ with $d(s,x) = k-1$ that is adjacent to all $u \in N(v)$
  such that $d(s,u) \leq k$. Note that for any $u \in N(v)$, we have
  $d(s,u) \leq k$ if and only if $L(u) \in
  \{\precc(L(v)),L(v)\}$. Consequently, this shows that \rMH holds.
\end{proof}

We now consider a graph $G$ and a labeling $L: V(G) \to \{0,1,2\}$
satisfying the rules \rMG and \rMGt. We show that if $G$ is a weakly bridged
graph and if $L$ satisfies \rMTC and \rMB, or if $G$ is a Helly graph
and $L$ satisfies \rMH, then there exists $s$ such that $L =
f_s$. Since bridged and Helly graphs have simply connected triangle
complexes, by \Cref{lem-acyclic}, the graph $\oG_L$ is acyclic and it
contains at least one source $s$. We now show that $L$ coincides with
the map $f_s$ when the graph is weakly bridged (respectively, Helly) and when
$L$ satisfies rules \rMG, \rMTC, and \rMB (respectively, \rMG and
\rMH). The following lemmas are the equivalent of
\Cref{lem-meshed-mod}.

\begin{lemma}\label{lem-bridged-tcb}
  Consider a graph $G$ and a labeling $L: V\rightarrow \{ 0,1,2\}$
  satisfying \rMG, \rMTC, and \rMB. For any sink $s$ of $\oG_L$ and any
  $v \in V(G)$, we have $L(v) = f_s(v) = d(s,v) \pmod{3}$.
\end{lemma}

\begin{proof}
  We prove the lemma by induction on $k=d(s,v)$. Since $s$ is a sink,
  $s$ has no neighbor labeled by $\precc(L(s))$ and thus by \rMG, we
  necessarily have $L(s) = 0$ and $L(u) = 1$ for any $u \in
  N(s)$. This shows that the property holds for $k =0$ and $k=1$.

  Assume that the assertion holds for all vertices at distance at most
  $k\geq 1 $ from $s$ and consider $v\in V(G)$ such that
  $d(s,v)=k+1\ge 2$.  Pick any $u\sim v$ with $d(s,u)=k$ and any
  $w\sim u$ with $d(s,w)=k-1$. By induction hypothesis,
  $L(u) = d(s,u) \pmod{3}$ and $L(w)=d(s,w) \pmod{3}$. This implies
  that $L(w)=\precc(L(u))$. By contradiction, assume that
  $L(v) \neq d(s,w) \pmod{3}$. This implies that either
  $L(v)=\precc(L(u))=L(w)$ or $L(v)=L(u)=\next(L(w))$. If
  $L(v) = L(w) = \precc(L(u))$, then by \rMB applied at $u$, we have
  $v \sim w$, but this is impossible as $d(s,v) = d(s,w)+2$. If
  $L(v) = L(u)$, then by \rMTC, there exists $x \sim u,v$ such that
  $L(x)=\precc(L(u))$. Note that $k \leq d(s,x) \leq k+1$. If
  $d(s,x) = k+1$, we get a contradiction by the previous case,
  replacing $v$ by $x$. Consequently, $d(s,x) = k = d(s,u)$ and by
  induction hypothesis, $L(x) = d(s,x) \pmod{3} = L(u)$, contradicting
  the fact that $L(x) = \precc(L(u))$. This concludes the proof that
  $L(v) = f_s(v) = d(s,v)\pmod{3}$.
\end{proof}

\begin{lemma}\label{lem-helly-mh}
  Consider a graph $G$ and a labeling $L: V\rightarrow \{ 0,1,2\}$
  satisfying \rMG and \rMH. For any sink $s$ of $\oG_L$ and any
  $v \in V(G)$, we have $L(v) = f_s(v) = d(s,v) \pmod{3}$.
\end{lemma}

\begin{proof}
  We prove the lemma by induction on $k=d(s,v)$. Since $s$ is a sink,
  $s$ has no neighbor labeled by $\precc(L(s))$ and thus by \rMG, we
  necessarily have $L(s) = 0$ and $L(u) = 1$ for any $u \in
  N(s)$. This shows that the property holds for $k =0$ and $k=1$.

  Assume that the assertion holds for all vertices at distance at most
  $k\geq 1 $ from $s$ and consider $v\in V(G)$ such that
  $d(s,v)=k+1\ge 2$.  Pick any $u\sim v$ with $d(s,u)=k$ and any
  $w\sim u$ with $d(s,w)=k-1$. By induction hypothesis,
  $L(u) = d(s,u) \pmod{3}$ and $L(w)=d(s,w) \pmod{3}$.  This implies
  that $L(w)=\precc(L(u))$. By contradiction, assume that
  $L(v) \neq d(s,w) \pmod{3}$. This implies that
  $L(v) \in \{L(u),\precc(L(u))\}$. Since $w$ is a neighbor of $u$
  with $L(w) = \precc(L(u))$, by \rMH applied at $u$, there exists
  $x \sim u$ with $L(x) = \precc(L(u)) = L(w)$ such that $x\sim
  w,v$. Since $d(s,v) = k+1 = d(s,w) +2$, we have
  $d(s,x) =k = d(s,u)$. By induction hypothesis, we then have
  $L(x) = k \pmod{3} = L(u)$, but this contradicts the fact that
  $L(x) = \precc(L(u))$.  This concludes the proof that
  $L(v) = f_s(v) = d(s,v)\pmod{3}$.
\end{proof}

We can now prove Theorem~\ref{th-bridged-helly-election}

\begin{proof}[Proof of Theorem~\ref{th-bridged-helly-election}]
  The fact that in any weakly bridged graph $G$ (respectively, Helly graph
  $G$), for any root $s$, the map $f_s$ satisfies \rMG, \rMGt, \rMTC,
  \rMB (respectively \rMG, \rMGt, \rMH) follows from
  Lemma~\ref{lem-gen-fr-ok} and Lemma~\ref{lem-fr-bridged-ok}
  (resp. Lemma~\ref{lem-fr-helly-ok}).

  Consider now a weakly bridged graph $G$ (respectively, a Helly graph $G$) and a
  labeling $L: V(G) \to \{0,1,2\}$ satisfying the rules  \rMG, \rMGt,
  \rMTC, \rMB (respectively, \rMG, \rMGt, \rMH) at every $v \in
  V(G)$. Since the triangle complex $X\tr(G)$ is simply
  connected, by \Cref{lem-acyclic}, $\oG_L$ is acyclic and thus
  contains at least one sink. By Lemma~\ref{lem-bridged-tcb}
  (respectively, Lemma~\ref{lem-helly-mh}), for any sink $s$ of
  $\oG_L$ and any $v \in V(G)$, $L(v) = d(s,v) \pmod{3} = f_s(v)$. For
  the same reasons as in the proof of \Cref{meshed-modulo3}, this
  implies that $\oG_L$ has a unique sink.  Designating the sink as a
  leader and every non-sink vertex as a non-leader, we thus have
  designed a local leader election protocol for weakly bridged graphs
  (respectively, Helly graphs) using only the three labels $0,1,2$.

  Note that these protocols are $1$-local. Indeed, for any
  $v \in V(G)$, we can check whether the rules \rMG, \rMGt, \rMTC,
  \rMB, and \rMH hold by considering only the labeled ball $B_1(v)$.
\end{proof}

\section{Local recognition}\label{sec-recognition}
The goal of this section, is to prove the following local recognition
result:

\begin{theorem}\label{local-recogn} Each of the following
  classes of meshed graphs admits an $r$-local recognition protocol
  with labels of size $O(\log D)$:
  \begin{enumerate}[(1)]
  \item chordal and  bridged graphs  with $r=1$,
  \item weakly bridged graphs and Helly graphs with $r=2$,
  \item weakly modular, modular, median, pseudo-modular, bucolic, dual
    polar, sweakly modular graphs with $r=3$,
  \item basis graphs of matroids and of even $\Delta$-matroids with
    $r=3$.
  \end{enumerate}
\end{theorem}

Consider a graph $G$ endowed with a labeling $D: V(G) \to
\N$. Consider the 2-local verifier $\cA_M$ that accepts at a vertex $v$
if the rules \rDGz, \rDG, \rDTC, and \rDM are satisfied at
$v$. Similarly, let $\cA_B$ be the $1$-local verifier that accepts at
a vertex $v$ if the rules \rDGz, \rDG, \rDTC, and \rDB are satisfied at
$v$. Finally, let $\cA_H$ be the $1$-local verifier that accepts at a
vertex $v$ if the rules \rDGz, \rDG, \rDTC, and \rDH are satisfied at
$v$.  We show that $\cA_M$ accepts all meshed graphs (as well as
potentially many other graphs). Similarly, $\cA_B$ accepts all weakly
bridged graphs and $\cA_H$ accepts all Helly graphs.

\begin{lemma}\label{lem-recog-ds}
  \begin{enumerate}[(1)]
  \item For any meshed graph $G$ and any root $s \in V(G)$, if
  $D(v) = d(s,v)$ for every $v \in V(G)$, then the verifier $\cA_M$
  accepts at every vertex $v \in V(G)$.

  \item For any weakly bridged graph $G$, and any root
  $s \in V(G)$, if $D(v) = d(s,v)$ for every $v \in V(G)$, then the
  verifier $\cA_B$ accepts at every vertex $v \in V(G)$. 

  \item For any Helly graph $G$, and any root $s \in V(G)$, if
  $D(v) = d(s,v)$ for every $v \in V(G)$, then the verifier $\cA_H$
  accepts at every vertex $v \in V(G)$.
  \end{enumerate}
\end{lemma}

\begin{proof}
  Consider a meshed graph $G$ and a root $s \in V(G)$.  The fact that
  $D=d(s,\cdot)$ satisfies \rDGz, \rDG, and \rDM at every vertex $v$
  follows from Theorem~\ref{th-meshed-distance}. The fact that \rDTC
  holds at every vertex $v$ follows from the fact that meshed graphs
  satisfy TC$(s)$.

  From Theorem~\ref{th-bridged-distance}, for any weakly bridged graph
  $G$ and for any root $s \in V(G)$, the map $D=d(s,\cdot)$ satisfies
  \rDGz, \rDG, \rDTC and \rDB.

  From Theorem~\ref{th-helly-distance}, for any Helly graph
  $G$ and for any root $s \in V(G)$, the map $D=d(s,\cdot)$ satisfies
  \rDGz, \rDG, and \rDH. The fact that \rDTC holds at every vertex $v$
  follows from the fact that Helly graphs satisfy TC$(s)$.
\end{proof}

We now consider an arbitrary graph $G$ and a labeling function
$D: V(G) \to \N$. We show that if $\cA_M$ (respectively, $\cA_B$ or
$\cA_H$) accepts $G$, then the triangle-square complex $X\trsq(G)$
(respectively, the triangle complex $X\tr(G)$) is simply
connected. The proof of this lemma follows the strategy of the proof
of~\cite[Lemma~3.41]{CCHO} but we show here that we can replace the
distances $d(s,v)$ to a fixed root $s$ by the values $D(v)$.

\begin{lemma}\label{lem-recog-sc}
  
  \begin{enumerate}[(1)]
  \item For any graph $G$ and any labeling $D: V(G) \to \N$, if \rDGz,
    \rDG, \rDTC, and \rDM hold at every vertex $v \in V(G)$, then the
    triangle-square complex $X\trsq(G)$ is simply connected.
  \item For any graph $G$ and any labeling $D: V(G) \to \N$, if \rDGz,
    \rDG, \rDTC, and \rDB hold at every vertex $v \in V(G)$, then the
    triangle complex $X\tr(G)$ is simply connected.
  \item For any graph $G$ and any labeling $D: V(G) \to \N$, if \rDGz,
    \rDG, \rDTC, and \rDH hold at every vertex $v \in V(G)$, then the
    triangle complex $X\tr(G)$ is simply connected.
  \end{enumerate}
\end{lemma}

\begin{proof}
  Consider a graph $G$ and a labeling $D: V(G) \to \N$ such that
  \rDGz, \rDG, \rDTC, and \rDM (respectively, \rDB or \rDH) hold at every
  $v \in V(G)$ and assume, by contradiction, that $X\trsq(G)$
  (respectively, $X\tr(G)$) is not simply connected. Let $A$ be the
  non-empty set of cycles in $G$, which are not null-homotopic, For a
  cycle $C\in A$, let $r(C)$ denote $\max\{ D(v): v\in C\}$. By \rDGz
  we conclude that $r(C)>0$. Let $A'\subseteq A$ be the set of cycles
  $C$ with minimal $r(C)$ among cycles in $A$. Let $r:=r(C)$ for some
  $C\in A'$.  Let $A''\subseteq A'$ be the set of cycles having the
  minimal number $e$ of edges $uv$ with $D(u)=D(v)=r$.  Further, let
  $A_0\subseteq A''$ be the set of cycles with the minimal number $m$
  of vertices $v$ with $D(v)=r$. Consider a cycle
  $C = (v_0,v_1,\ldots,v_{k-1},v_0)$ in $A_0$. We can assume without
  loss of generality that $D(v_1)=r$. From the definition of $r$ and
  condition \rDG, we deduce that $r-1\le D(v_0)\le r$ and
  $r-1\le D(v_2)\le r$.  We distinguish two cases, depending on the
  values of $D(v_0)$ and $D(v_2)$.
  \begin{case}\label{case 1}
    $D(v_0)=r$ or $D(v_2)=r$. 
  \end{case}
  Assume without loss of generality that $D(v_0)=r$. Then, by \rDTC
  there exists a vertex $x\sim v_0,v_1$ with $D(x)=r-1$.  Observe that
  the cycle $C'=(v_0,x,v_1,\ldots,v_{k-1},v_1)$ belongs to $A'$
  because in $X\trsq(G)$ it is homotopic to $C$ by a homotopy going
  through the triangle $xv_0v_1$. The number of edges $uv$ of $C'$
  such that $D(u)=D(v)=r$ is less than $e$ (we removed the edge
  $v_0v_1$). This contradicts the choice of the number $e$.

  \begin{case}\label{case2}
    $D(v_0)=D(v_2)=r-1$.
  \end{case}
  If $v_0 \sim v_2$, then $C$ is homotopic to
  $C'=(v_0,v_2,\ldots,v_{k-1},v_1)$ by a homotopy going through the
  triangle $v_0v_1v_2$ and therefore, $C' \in A''$. Observe that $C'$
  has fewer vertices with $D(v) = r$ than $C$ (we removed $v_1$). This
  contradicts the choice of the number $m$. In the following, we thus
  assume that $v_0 \nsim v_2$.  If \rDB holds at $v_1$, we immediately
  get a contradiction.

  Suppose now that \rDH holds at $v_1$. Then there exists a vertex
  $y\sim v_0,v_1,v_2$ with $D(y) = r-1$. Consequently, the cycle
  $C''=(v_0,y,v_2,...,v_k,v_1)$ is homotopic to $C$ (via the triangles
  $v_0yv_1$ and $v_1yv_2$). Thus $C''$ belongs to $A''$ and the number
  of its vertices $u$ with $D(u)=r$ is equal to $m-1$. This
  contradicts the choice of the number $m$.
  
  Suppose now that \rDM holds at $v_1$. Then, there exists a vertex
  $y\sim v_0,v_2$ with $D(y)\le r-1$. If $y \sim v_1$, then we get the
  same contradiction as in the case where \rDH holds at $v_1$. If
  $y \nsim v_1$, the cycle $C''=(v_0,y,v_2,...,v_k,v_1)$ is also
  homotopic to $C$ (via the square $v_0v_1v_2y$). As before, $C''$
  belongs to $A''$ and this contradicts the choice of the number $m$.

  In all the cases above we get a contradiction. Consequently, it follows that the
  set $A$ is empty and hence the lemma holds.
\end{proof}

As a consequence of Lemmas~\ref{lem-recog-ds} and~\ref{lem-recog-sc},
there exists a proof labeling scheme that accepts all meshed graphs
(respectively, all weakly bridged graphs and all Helly graphs) and rejects all graphs that
do not have a simply connected triangle-square complex (respectively,
triangle complex). This allows us to use the existing local-to-global
characterizations in order to provide local recognition protocols for
the subclasses of meshed graphs mentioned in
Theorem~\ref{local-recogn}. We establish the different cases of
Theorem~\ref{local-recogn} in the following subsections. In each case,
we recall the necessary definitions and characterizations and then
show how to use them to design a local recognition algorithm.

\subsection{1-Local recognition of bridged and chordal graphs}

By~\cite[Theorem 8.1]{Ch_CAT}, a graph $G$ is bridged if and only if
the triangle complex $X\tr(G)$ of $G$ is simply connected and $G$ does
not contain induced 4-wheels $W_4$ and 5-wheels $W_5$. Clearly, there
exists a $1$-local rule that accepts a graph $G$ if $G$ does not
contain any induced $W_4$ or $W_5$. Then, a $1$-local recognition
protocol for bridged graphs can be obtained by adding this 1-local
rule to the verifier $\cA_B$.

Indeed, for any graph $G$, if $G$ is bridged, by
Lemma~\ref{lem-recog-ds}, for any root $s \in V(G)$, when endowed with
the labeling $D = d(s,\cdot)$, $\cA_B$ accepts at every vertex
$v$. Since $G$ is bridged, $G$ does not contain any $W_4$ or $W_5$ and
$G$ is accepted.  Conversely, if $G$ is not bridged, then either $G$
contains an induced $W_4$ or $W_5$, or $X\tr(G)$ is not simply
connected. In the first case, the verifier rejects at the vertex $u$
that is the center of the wheel $W_4$ or $W_5$.  In the second case,
by Lemma~\ref{lem-recog-sc}, for any labeling $D: V(G) \to \N$,
$\cA_B$ rejects $G$ at some vertex $v$.

We now show that chordal graphs are precisely the bridged graphs not
containing induced $k$-wheels $W_k$ with $k\ge 4$. As detecting wheels
can be done by considering $1$-balls, this implies that there exists a
$1$-local recognition protocol for chordal graphs.

\begin{lemma}
  Chordal graphs are precisely the bridged graphs not containing
  induced $k$-wheels $W_k$ with $k\ge 4$. Therefore, they are the
  graphs with simply connected triangle complexes not containing
  induced $k$-wheels $W_k$ with $k\ge 4$.
\end{lemma}

\begin{proof}
  The condition is obviously necessary. Conversely, consider a bridged
  graph that does not contain any induced $W_k$ with $k\ge 4$ and
  suppose it is not chordal. Let $k \geq 4$ be the minimal length of
  an induced cycle that is not a triangle.  Consider an induced cycle
  $C = (v_0, \ldots, v_{k-1})$ of length $\ell(C) = k$ that minimizes
  $D(C)=\sum_{0\leq i \leq k-1}d(v_0,v_i)$. Let
  $p = \max \{d(v_0,v_i): 0\leq i \leq k-1\}$, and observe that
  $p \geq 2$ as $v_0 \nsim v_2$.  Consider a minimal index $j$ such
  that $d(v_0,v_j) = p$. Then $d(v_0,v_{j-1}) = p-1$ and
  $p-1 \leq d(v_0,v_{j+1}) \leq p$. By the convexity of $B_{p-1}(v_0)$
  and since $v_{j-1} \nsim v_{j+1}$, we necessarily have
  $d(v_0,v_{j+1}) =p$. By TC($v_0$), there exists $c \sim v_j,v_{j+1}$
  such that $d(v_0,c) = p-1$. By the convexity of $B_{p-1}(v_0)$,
  $c \sim v_{j-1}$.  Consider the cycle
  $C' = (v_0, \ldots, v_{j-1},c,v_{j+1}, \ldots, v_{p-1})$ obtained by
  replacing $v_i$ by $c$ in $C$ and observe that $\ell(C') = k$ and
  $D(C') = D(C) -1$. By our choice of $C$, the cycle $C'$ cannot be
  induced, and by the definition of $k$ and since $C$ is induced, we
  deduce that $c\sim v_i$, for all $0 \leq i \leq k-1$. Consequently,
  $c$ and the vertices of $C$ induce a wheel $W_k$, contradicting our
  initial hypothesis.

  The second statement follows from the local-to-global
  characterization of bridged graphs~\cite{Ch_CAT}.
\end{proof}

\subsection{2-local recognition of weakly bridged graphs and Helly
  graphs}

Recall that weakly bridged graphs are the weakly modular graphs not
containing squares \cite{ChOs}.  An \emph{extended 5-wheel}
$\widehat{W}_5$ is a 5-wheel $W_5$ plus a triangle sharing exactly one
edge with $W_5$.  A graph $G$ satisfies the
\emph{$\widehat{W}_5$-wheel condition} if for any induced $\widehat{W}_5$
there exists a vertex adjacent to all vertices of $\widehat{W}_5$. By
\cite[Theorem A]{ChOs}, weakly bridged graphs are exactly the graphs
$G$ whose triangle complexes $X\tr(G)$ are simply connected, that do
not contain induced cycles of length 4, and that satisfy the
$\widehat{W}_5$-condition.
Clearly, we can detect if a graph $G$ contains a cycle of length $4$
by considering the $2$-balls of $G$. Similarly, any extended $5$-wheel
is contained in the $2$-ball of the center $c$ of the wheel, and one
can thus verify if a graph satisfies the $\widehat{W}_5$-wheel
condition by considering the $2$-balls of $G$. Consequently, a
$2$-local recognition protocol for weakly bridged graphs can be
obtained by adding to the verifier $\cA_B$ the 2-local rules checking
that $G$ does not contain an induced $C_4$ and satisfies the
$\widehat{W}_5$-wheel condition. The proof of correctness works along
the same lines as the proof for bridged graphs.

Recall that a graph $G$ is a \emph{Helly graph} if any collection of
pairwise intersecting balls of $G$ has a nonempty
intersection.
Analogously, $G$ is a \emph{clique-Helly graph} if any collection of pairwise
intersecting maximal cliques of $G$ has a nonempty intersection.
By \cite[Theorem 3.8]{CCHO} Helly graphs are exactly the clique-Helly
graphs $G$ with simply connected triangle complexes
$X\tr(G)$. Clique-Hellyness is clearly a local property. In fact, we can use the
following characterization of clique-Helly graphs~\cite{Dr} (which is
the specification to cliques of a result of \cite{BeDu} for Helly
hypergraphs): a graph $G$ is clique-Helly if and only if for any
triangle $T=uvw$ of $G$ the set $T^*$ of all vertices of $G$ adjacent
with at least two vertices of $T$ contains a universal vertex, i.e., a
vertex adjacent to all remaining vertices of $T^*$. Since
$T^*\subseteq B_2(u)\cap B_2(v)\cap B_2(w)$, the construction of $T^*$
and the existence of the universal vertex for $T^*$ can be ensured by
a 2-local verifier.  Consequently, a $2$-local recognition protocol
for Helly graphs can be obtained by adding the rules of this $2$-local
verifier for clique-Helly graphs to the verifier $\cA_H$. The proof of
correctness works along the same lines as the proof for bridged graphs.

\subsection{3-Local recognition protocols for weakly modular, modular,
  median, sweakly modular, dual polar, bucolic, cage-amalgation, and
  pseudo-modular graphs.}
A graph $G$ is called \emph{locally weakly modular} if it satisfies the \emph{local triangle condition (LTC)} and the \emph{local quadrangle condition (LQC)}. (LTC) asserts that for any triplet of vertices $u,v,w$ such that $v\sim w$ and $d(u, v) = d(u, w) = 2$ there exists $x$ adjacent to $u,v,w$. (LQC) asserts that for any four vertices $u,v,w,z$ such
that $z\sim v,w$ and $d(v,w)=d(u,v)=d(u,w)=d(u,z)-1=2$, there
exists a common neighbor $x$ of $u, v$ and $w$. Clearly, (LTC) is a 2-local rule and (LQC) is a 3-local rule. 
By \cite[Theorem 3.1]{CCHO}, a graph $G$ is weakly modular if and only
if its triangle-square complex $X\trsq(G)$ is simply connected and $G$
is locally weakly modular.  Therefore, the local recognition protocol
for weakly modular graph is obtained by adding the 2-local rule (LTC)
and the 3-local rule to the verifier $\cA_M$.  The proof of
correctness works along the same lines as the proof for bridged graphs

We now consider several subclasses of weakly modular graphs and we
recall their definitions or their characterizations that allow to
derive local recognition protocols. 

\begin{enumerate}[(1)]
\item Modular graphs are precisely the triangle-free weakly modular
  graphs~\cite{Ch_metric},
\item Median graphs are modular graphs without induced
  $K_{2,3}$~\cite{Mu},
\item Sweakly modular graphs are weakly modular graphs that do not
  contain $K_4^-$ or $K_{3,3}^-$ as induced subgraphs~\cite{CCHO},
\item Dual polar graphs are thick sweakly modular graphs, i.e.,
  sweakly modular graphs where any pair of vertices $u,v$ at distance
  $2$ are contained in a common square~\cite{Ca,CCHO},
\item Bucolic graphs are the weakly modular graphs that do not contain
  $K_{2,3}, W_4^-$, or $W_4$ as induced subgraphs~\cite{BrChChGoOs},
\item Cage-amalgamation graphs are the the bucolic graphs that do not
  contain $W_k$ for $k \geq 4$ as induced subgraphs~\cite{BrChChKoLaVa},
\item Pseudo-modular graphs are weakly modular graphs that do not
  contain metric triangle size 2~\cite{Ch_metric}.
\end{enumerate}

By inspecting balls of radius at most 3 around each
vertex, one can decide if a graph contains a triangle, a $K_{2,3}$, a
$K_4^-$, a $K_{3,3}^-$, a $W_4^-$, or a $W_k$, $k \geq 4$. Note that
we can also check if a graph is thick by considering each ball of
radius 2. Finally, since metric triangles of weakly modular graphs are
strongly equilateral, each metric triangle $v_1v_2v_3$ of size 2 in a
weakly modular graph $G$ is contained in the ball $B_2(v_1)$.  This
shows that the 3-local recognition protocol for weakly modular graphs
can be easily transformed in 3-local recognition protocols for all
these classes.

\subsection{3-Local recognition of basis graphs of matroids and even
  $\Delta$-matroids}

A \emph{matroid} on a
finite set $I$ is a collection $\cB\subseteq 2^I$, called
\emph{bases}, which satisfy the following exchange property: for all
$A, B\in \cB$ and $a\in A\setminus B$ there exists
$b\in B\setminus A$ such that
$A\setminus \{ a\}\cup \{ b\}\in \cB$ (the base
$A\setminus \{ a\}\cup \{ b\}$ is obtained from the base A by an
\emph{exchange}). All the bases of a matroid have the same
cardinality. 
$\Delta$-matroids have been introduced independently in~\cite{Bou,ChKa,DrHa}.
A $\Delta$–matroid is a collection $\cB\subseteq 2^I$, called \emph{bases} (not necessarily
equicardinal) satisfying the symmetric exchange property: for any
$A,B\in \cB$ and $a\in A\Delta B$, there exists
$b\in B\Delta A$ such that $A\Delta \{ a,b\}\in \cB$.  A
$\Delta$–matroid whose bases all have the same cardinality modulo 2 is
called an \emph{even $\Delta$–matroid}. Obviously, matroids are even $\Delta$-matroids. The \emph{basis graph} $G = G(\cB)$ of a matroid or even $\Delta$–matroid $\cB$ is the graph whose vertices are the
bases of $\cB$ and edges are the pairs $A, B$ of bases
differing by a single exchange, i.e., $|A\Delta B|= 2$.

Basis graphs of matroids and even $\Delta$-matroids are 
meshed graphs~\cite{Ch_delta}. Basis graphs of matroids have been
nicely characterized in a metric way by Maurer \cite{Mau}. His result
was extended in \cite{Ch_delta} to basis graphs of even
$\Delta$-matroids.  Basis
graphs of matroids and even $\Delta$-matroids have been characterized
in a local-to-global way in \cite{ChChOs}. We now recall these
characterizations.  A graph $G$ satisfies the \emph{Positioning
  Condition (PC)} if for any vertex $s$ and any square $uvwx$ of $G$,
we have $d(s,u)+d(s,w)=d(s,v)+d(s,x)$. If (PC) holds for all squares
$uvwx$ belonging to the ball $B_3(s)$, then this  is called
the \emph{Local Positioning Condition (LPC)}. A graph satisfies the
\emph{Interval Condition (IC$_k$)} if for any two vertices $u,v$ at
distance 2, $u,v$ belongs to a common $4$-cycle (i.e., the graph is
thick) and the subgraph induced by $u$, $v$, and their common
neighbors induce an induced subgraph of the $k$-octahedron.  Finally,
a graph $G$ satisfies the \emph{Link Condition (LC)} (respectively,
the \emph{Bipartite Link Condition (BLC)}) if the neighborhood $N(v)$
of each vertex $v$ induce the line graph of a graph (respectively, the
line graph of a bipartite graph). 

According to~\cite{Mau},  $G$ is the basis graph of a matroid
iff $G$ satisfies (PC), (IC3), and (BLC). According
to~\cite{Ch_delta}, $G$ is the basis graph of an even $\Delta$-matroid
iff $G$ satisfies (PC), (IC4), and (LC). Observe that
(IC3), (IC4), (BLC), and (LC) can be locally verified, but (PC) is a
global condition. Answering to a conjecture of Maurer~\cite{Mau}, the
following local-to-global characterization of basis graphs was established in~\cite{ChChOs}: a graph $G$ is the basis graph of a
matroid (respectively, the basis graph of an even $\Delta$-matroid) iff
the triangle-square complex $X\trsq(G)$ of $G$ is simply connected and
$G$ satisfies (LPC) and (IC3) (respectively, (LPC), (IC4), and (LC)).

Since basis graphs of matroids (respectively, even $\Delta$-matroids)
are meshed and since (LPC) and (IC3) (respectively, (LPC), (IC4), and
(LC)) can be checked by a $3$-local verifier, we can obtain a
$3$-local recognition protocol for basis graphs of matroids
(respectively, even $\Delta$-matroids), by adding the rules of this
$2$-local verifier for local conditions to the verifier $\cA_M$. The
proof of correctness works along the same lines as the proof for
bridged graphs.

\section{Final remarks}

Meshed graphs rather constitute a general framework for numerous
classes of graphs than a specific class of graphs. The key property
that allows to prove \Cref{local-recogn} is that our
distance-certification protocol is able to ``separate'' meshed graphs
from graphs that do not have a simply connected triangle-square
complex. We would like to emphasize that no proof labeling scheme can
recognize graphs that have simply connected triangle-square complexes
because recognizing these graphs is undecidable~\cite{Hak}. On the
other hand, meshed graphs do not admit local-to-global
characterizations~\cite{CCHO}, and therefore, although all subclasses
of meshed graphs from Theorem~\ref{local-recogn} admit $r$-local
recognition protocols with $r\in \{1,2,3\}$, we believe that \emph{the
  class of all meshed graphs does not admit a $r$-local recognition
  protocol for any $r$}. If this assertion is true, then our proof
labeling scheme would allow to separate two natural classes of
graphs even if none of them has a local recognition protocol.

Our leader election protocol for meshed graphs is based on the
properties of the distance function in meshed graphs. These properties
enable us to ensure that even if we have only labels in $\{0, 1, 2\}$, we
can guarantee  that the labels given to the nodes are their distances to
some root counted modulo 3. The same properties enable to prove that
the triangle-square complex of any meshed graph is simply
connected. It is then natural to consider the problem in other classes
of graphs that have nice properties that guarantee that their
triangle-square complexes are simply connected, and to check if such
properties allow to design a local leader election protocol.
Dismantlable graphs~\cite{NoWi,Qui83} mentioned in the introduction is
one such class. They generalize bridged and Helly graphs. Not only their
clique complexes are simply connected, but they are
contractible. There exists a local leader election protocol for
$K_4$-free dismantlable graphs~\cite{ChKo}, but we do not know if it
can be extended to general dismantlable graphs. The distance-based approach of
Section~\ref{sec-election} seems difficult to extend, as there is no
distance-based characterization of dismantlable graphs. Another class
to consider is the class of \emph{ample/lopsided/extremal
  sets}~\cite{La,BaChDrKo,BoRa}. Ample sets generalize median graphs
and have simply connected square complexes (and even contractible cube
complexes), as well as many combinatorial and metric
characterizations and properties.

\subsection*{Acknowledgements} 
J.C. and V.C. were partially supported by the ANR project MIMETIQUE ``Mineurs métriques'' (ANR-25-CE48-4089-01).

% \bibliographystyle{plain}
% \bibliography{biblio}

\end{document}